\documentclass{article}

\usepackage[margin=3cm]{geometry}
\usepackage[utf8]{inputenc}

\usepackage{amsmath}
\usepackage{amssymb}
\usepackage{amsthm}
\usepackage{tikz}
\usepackage{subcaption}
\usepackage{enumitem}

\newtheorem{lemma}{Lemma}

\theoremstyle{definition}
\newtheorem{definition}{Definition}

\usepackage{thmtools}

\usepackage{todonotes}

\makeatletter
\renewcommand*\env@matrix[1][\arraystretch]{%
  \edef\arraystretch{#1}%
  \hskip -\arraycolsep
  \let\@ifnextchar\new@ifnextchar
  \array{*\c@MaxMatrixCols c}}
\makeatother

\newcommand{\round}{\textup{\textsf{round}}}
\newcommand{\var}{\mathit{var}}
\newcommand{\reduce}{\textup{\textsf{reduce}}}
\newcommand{\union}{\textup{\textsf{union}}}

\newcommand{\estimateAndSample}{\textup{\textsf{estimateAndSample}}}
\newcommand{\approxMCnFBDDcore}{\textup{\textsf{approxMCnFBDD\_core}}}
\newcommand{\approxMCnFBDD}{\textup{\textsf{approxMCnFBDD}}}

\newcommand{\children}{\mathsf{children}}
\newcommand{\descendants}{\mathsf{desc}}
\newcommand{\Ex}{\textup{E}}
\newcommand{\Va}{\textsf{Var}}

\newcommand{\calP}{\mathcal{P}}

\newcommand{\calH}{\mathcal{H}}

\newcommand{\calS}{\mathcal{S}}

\newcommand{\paths}{\mathsf{path}}
\newcommand{\lcpn}{\mathsf{lcpn}}
\newcommand{\median}{\mathsf{median}}
\newcommand{\mods}{\mathsf{mod}}

\newcommand{\vY}[4]{\mathfrak{S}^{#1}_{#2,#3}(#4)}
\newcommand{\vZ}[3]{\hat{\mathfrak{S}}^{#1}_{#2}(#3)}

\newcommand{\pUnion}{sample\_union}
\newcommand{\pEstimate}{compute\_estimate}
\newcommand{\pReduce}{sample\_reduce}

\usepackage{url}
\usepackage{authblk}
\usepackage{thm-restate}  
  
\usepackage[linesnumbered,lined,ruled,vlined]{algorithm2e}
\usepackage{tikz}
\usepackage{bbm}

\title{An FPRAS for Model Counting for Non-Deterministic Read-Once Branching Programs \thanks{
	The authors decided to forgo the old convention of alphabetical
	ordering of authors in favor of a randomized ordering, denoted by \textcircled{r}. The publicly verifiable record of the randomization is available
	at \protect\url{https://www.aeaweb.org/journals/policies/random-author-order/search}
}}
\date{}
\author{Kuldeep S. Meel$^1$ \textcircled{r} Alexis de Colnet$^2$\\ \medspace
	
	$^1$University of Toronto, Canada\\
	$^2$TU Wien, Austria
}

\begin{document}

\maketitle

\begin{abstract}
Non-deterministic read-once branching programs, also known as non-deterministic free binary decision diagrams (nFBDD), are a fundamental data structure in computer science for representing Boolean functions. In this paper, we focus on \#nFBDD, the problem of model counting for non-deterministic read-once branching programs. The \#nFBDD problem is \#P-hard, and it is known that there exists a quasi-polynomial randomized approximation scheme for \#nFBDD. In this paper, we provide the first FPRAS for \#nFBDD. Our result relies on the introduction of new analysis techniques that focus on bounding the dependence of samples.
\end{abstract}

\section{Introduction}\label{section:introduction}

Read-once branching programs or binary decision diagrams are fundamental data structures in computer science used to represent Boolean functions. Their variants have been discovered multiple times across various sub-fields of computer science, and consequently, they are referred to by many acronyms~\cite{Akers78,Bryant86,Wegener00}. In this paper, we focus on non-deterministic read-once branching programs, also known as non-deterministic free binary decision diagrams (nFBDD).

We study the following computational problem:

\noindent {\bfseries \#nFBDD}: Given a non-deterministic read-once branching program $B$ over a Boolean set of variables $X$, compute the number of models of $B$, i.e., the number of assignments over $X$ that $B$ maps to $1$.

From a database perspective, \#nFBDD is an important problem owing to the recent connections between query evaluation and knowledge compilation~\cite{JhaS12,JhaS13,MonetO18,Monet20,AmarilliC24,amarilli2024circus}. The field of knowledge compilation has its origins in the artificial intelligence community, where functions represented in input languages are compiled into target languages that can support  queries tractably (often viewed as polynomial time)~\cite{DarwicheM02}. The typical queries of interest are  satisfiability, entailment,  enumeration, and counting. 

The target languages in the context of databases have been variants of binary decision diagrams, also referred to as branching programs, and circuits in decomposable negation normal form (DNNF)~\cite{AmarilliC24,amarilli2024circus}. A binary decision diagram is a representation of a Boolean function as a directed acyclic graphs where the nodes correspond to variables and the sinks correspond to values, i.e., $0$ or $1$. One of the most well-studied forms is the ordered binary decision diagram (OBDD), where the nodes correspond to variables and, along every path from root to leaf, the variables appear in the same order~\cite{Bryant86}. A generalization of OBDD is nOBDD, where internal nodes can also represent disjunction ($\vee$) gates. 

nFBDD are a natural generalization of nOBDDs, as they do not impose restrictions on the ordering of Boolean variables. Since nFBDD do not impose such restrictions, they are known to be exponentially more succinct than nOBDD; that is, there exist functions for which the smallest nOBDD is exponentially larger than the smallest nFBDD~\cite{AmarilliCMS20}. From this viewpoint, nFBDDs occupy a space between nOBDD and DNNF circuits, as they are exponentially more succinct than nOBDDs, while DNNFs are quasi-polynomially more succinct than nFBDD~\cite{BL15,AmarilliCMS20}.

 In the context of databases, the connection between knowledge compilation and query evaluation has been fruitful, leading to the discovery of both tractable algorithms and lower bounds. Of particular note is the application of the knowledge compilation paradigm in query evaluation on probabilistic databases~\cite{JhaS13}, Shapley value computation~\cite{DeutchFKM22},  the enumeration of query answers, probabilistic graph homomorphism~\cite{AvBM24}, counting answers to queries~\cite{ACJR21}. The knowledge compilation-based approach involves first representing the database task as a query over a Boolean function and then demonstrating that the corresponding Boolean function has a tractable representation in a given target language, which also supports the corresponding query in polynomial time~\cite{AC24}.
 For example, in the context of query evaluation over probabilistic databases, one can show that the problem of probabilistic query evaluation can be represented as a weighted model counting problem over nOBDD when the underlying query is a path query~\cite{AvBM24}. Since there is a fully polynomial-time randomized approximation scheme (FPRAS) for the problem of model counting over nOBDD~\cite{ACJR19}, it follows that probabilistic query evaluation for regular path queries over tuple-independent databases admits an FPRAS~\cite{AvBM24}.  In the context of aggregation tasks, the underlying query is often model counting and its variants~\cite{mengel2021counting}. 
 
 The aforementioned strategy makes it desirable to have target languages that are as expressive as possible while still supporting queries such as counting in polynomial time. In this context, the recent flurry of results has been enabled by the breakthrough of Arenas, Croquevielle, Jayaram, and Riveros, who showed that the problem of \#nOBDD admits an FPRAS~\cite{ACJR19}. As mentioned earlier, nOBDD imposes a severe restriction on variable ordering, i.e., along every path from root to leaf, the variable ordering remains the same. nFBDD generalizes nOBDD by alleviating this restriction, thereby enabling succinct representations for several functions that require exponentially large nOBDD. Since nFBDD generalize nOBDD, the \#P-hardness of  \#nFBDD, the problem of model counting over nFBDD, immediately follows.  Accordingly,   in light of the recent discovery of FPRAS for \#nOBDD, an important open question is whether there exists an FPRAS for \#nFBDD. The best known prior result provides a quasi-polynomial time algorithm owing to reduction of nFBDD to DNNF (and accordingly $(+,\times)$-programs)~\cite{GJKSM97}.  As noted in Section~\ref{sec:relatedwork},  the techniques developed in the context of design of FPRAS for \#nOBDD do not extend to handle the case for \#nFBDD and therefore, design of FPRAS for \#nFBDD would require development of new techniques.

 The primary contribution of our work is to answer the aforementioned question affirmatively, which is formalized in the following theorem.

\begin{restatable}{theorem}{mainResult}\label{theorem:main_result}
	Let $B$ be an nFBDD over $n$ variables, $\varepsilon > 0$ and $\delta > 0$. 
	Algorithm $\approxMCnFBDD(B,\varepsilon,\delta)$ runs in time $O(n^{11}|B|^{10}\varepsilon^{-4}\log(\delta^{-1}))$ and returns $\mathsf{est}$ with the guarantee that 
	$$
		\Pr\left[\mathsf{est}  \in (1 \pm \varepsilon) |B^{-1}(1)| \right] \geq 1 - \delta.
	$$
\end{restatable}

\paragraph*{Organization of the paper.} We start with background on nFBDD in Section~\ref{section:background}. The different components of the FPRAS are described in Section~\ref{section:algorithm} and the analysis is split in three parts: in Section~\ref{section:derivation_path} we introduce the key concept of derivation paths, in Section~\ref{section:random_process} we describe the particular framework for the analysis, and in Section~\ref{section:analysis} we go through the proof of the FPRAS guarantees.

\section{Background}\label{section:background}

Given a positive integer $n$ and $m$ an integer less than $n$, $[n]$ denotes the set $\{1,2,\dots,n\}$ and $[m,n]$ the set  $\{m,m+1,\dots,n\}$. For $a$, $b$ and $\varepsilon$ three real numbers with $\varepsilon > 0$, we use $a  \in (1 \pm \varepsilon)b$ to denote $(1 - \varepsilon)b \leq a \leq  (1 + \varepsilon)b$, similarly, $a  \in \frac{b}{1 \pm \varepsilon}$ stands for $\frac{b}{1 + \varepsilon}\leq a \leq  \frac{b}{1 - \varepsilon}$. We sometimes uses the special value $\infty$ with the convention that $\frac{1}{\infty}$ equals $0$ and that $\frac{1}{0}$ equals $\infty$.

Boolean variables take value $0$ (\emph{false}) or $1$ (\emph{true}). An assignment $\alpha$ to a set $X$ of Boolean variables is mapping from $X$ to $\{0,1\}$. We sometimes see $\alpha$ as a set $\{x \mapsto \alpha(x) \mid x \in X\}$. We denote by $\alpha_\emptyset$ the empty assignment, which corresponds to the empty set. The set of assignments to $X$ is denoted $\{0,1\}^X$. A Boolean function $f$ over $X$ is a mapping $\{0,1\}^X$ to $\{0,1\}$. The \emph{models} of $f$ are the assignments mapped to $1$ by $f$. when not explicit, the set of variables assigned by $\alpha$ is denoted by $var(\alpha)$. When $var(\alpha') \cap var(\alpha) = \emptyset$, we denote by $\alpha \cup \alpha'$ the assignment to $var(\alpha') \cup var(\alpha)$ consistent with both $\alpha$ and $\alpha'$. For $S$ and $S'$ two sets of assignments, we write $S \otimes S' = \{ \alpha \cup \alpha' \mid \alpha \in S,\, \alpha' \in S'\}$.

\paragraph*{nBDD.}

A \emph{binary decision diagram} (BDD) is a directed acyclic graph (DAG) with a single source node $q_\text{source}$, two \emph{sinks} labeled $0$ and $1$, and where each internal node is labeled by a Boolean variable $x$ and has two outgoing edges: the $0$-edge going to the $0$-child $q_0$ and and the $1$-edge going to the $1$-child $q_1$. Internal nodes are called \emph{decision nodes} and are written $ite(x,q_1,c_0)$ (if $x$ then $q_1$ else $q_0$). Note that $q_0$ and $q_1$ may be equal. A path in the DAG contains a variable $x$ if it contains a node $ite(x,q_1,q_0)$. Every variables assignment $\alpha$ corresponds to the unique path that starts from the source and, on a decision node $ite(x,q_1,q_0)$ follows the $\alpha(x)$-edge. \emph{Non-deterministic} BDD (nBDD) also admit \emph{guess nodes}: unlabeled nodes with arbitrarily many children. When a path for an assignment reaches a guess node it can pursue to any child, so several paths can correspond to the same assignment in an nBDD. For $q$ a node in an nBDD $B$, $var(q)$ denotes the set of variables labeling decision nodes reachable from $q$ (including $q$). We note $var(B) = var(q_\text{source})$. $B$ computes a Boolean function over $var(B)$ whose models are the assignments for which at least one path reaches the $1$-sink. Every node $q$ of $B$ is itself the source of an nBDD and therefore represents a function over $var(q)$ whose set of models we note $\mods(q)$. So $B^{-1}(1) = \mods(q_\text{source})$. The function computed by an nBDD is also that computed by the circuit obtained replacing every decision node $ite(x,q_1,q_0)$ by $(\neg x \land q_0) \lor (x \land q_1)$ and every guess node with children $q_1,\dots,q_k$ by $q_1 \lor \dots \lor q_k$. Thus, we call $\lor$-nodes the guess nodes in this paper. The size of an nBDD $B$, denoted by $|B|$, is its number of edges. 

\paragraph*{nFBDD.}

An nBDD is \emph{free} (nFBDD) when every path contains every variable at most once. There are also called in the literature \emph{read-once non-deterministic branching programs} (1-NBP). Note that in an nFBDD, variables may appear in different order depending on the path. When the order of occurrence of the variables is consistent among all paths of the nFBDD we say that we have an \emph{ordered} nBDD (nOBDD). We call an nFBDD \emph{$1$-complete} when along every path from the source to the $1$-sink, every variable occurs exactly once. We call an nFBDD \emph{$0$-reduced} when it contains no decision nodes $ite(x,0\text{-sink},0\text{-sink})$ and no $\lor$-nodes that have the $0$-sink among their children. Technically, $0$-reduced nFBDD cannot represent functions with no models, but these functions are not considered in this paper. An nFBDD is \emph{alternating} when its source node is a $\lor$-node, and when every $\lor$-node only has decision nodes for children, and when every decision node has only $\lor$-nodes and sinks for children. Every $n$-variables nFBDD $B$ can be made $1$-complete, $0$-reduced and alternating in polynomial time without modifying the function it computes just by adding nodes of the form $ite(x,q,q)$ or $\lor$-nodes with a single child. 

\begin{restatable}{lemma}{makeComplete}\label{lemma:makeComplete}
Every nFBDD $B$ over $n$ variables can be made $1$-complete, $0$-reduced, and alternating in time $O(n|B|^2)$.
\end{restatable}
\begin{proof}[Proof sketch]
First, we make $B$ $0$-reduced. To this end, repeat the following until reaching a fixed point: replace all nodes $ite(x,0\text{-sink},0\text{-sink})$ by the $0$-sink, remove the $0$-sink from $\lor$-nodes' children, and replace all $\lor$-nodes with no child by the $0$-sink. Doing these changes bottom-up in $B$ takes time $O(|B|)$ and results in a $0$-reduced nFBDD $B'$ with $|B'| \leq |B|$.

Second, we make $B'$ alternating with respect to the $\lor$-nodes. First, replace every $\lor$-node that have the $1$-sink as a child by the $1$-sink. After that, no sink is a child of any $\lor$-node. Next, for every $\lor$-node $q$ in $B'$, if $q$ has a parent $q'$ that is a $\lor$-node, then remove $q$ from $q'$'s children and add all of $q$'s children to $q'$'s children. Doing these changes bottom-up yields an nFBDD $B''$ whose $\lor$-nodes all have only decision nodes as children. We remove from $B''$ nodes that have no parent and are not the source node $B'$. $B''$ does not have more number of nodes than $B'$ but can have more edges. The number of children of each $\lor$-node is increased by at most $|B'|$ so constructing $B''$ takes time is $O(|B'|^2) = O(|B|^2)$. $B''$ is still $0$-reduced.

Third, we make $B'$ $1$-complete. For every node $q \in B''$, let $\children(q) = (q_1,\dots,q_k)$. While there exists $i \in [k]$ such that $var(q_i) \neq var(q)$, choose $x \in var(q) \setminus var(q_i)$ and replace $q_i$ by the node $ite(x,q_i,q_i)$ in the children of $q$. This operation adds at most $n$ decision nodes per original child of $q$, so it finished in time $O(n|B''|) = O(n|B|^2)$ when done for all decision nodes and gives a $1$-complete nFBDD $B'''$. Clearly $B'''$ is still $0$-reduced and alternating with respect to the $\lor$-nodes.

Finally we make $B'''$ alternating. For that we just have to consider every $ite(x,q_1,q_0)$ and, if $q_b$ is not a $\lor$-node or a sink, to replace it by a $\lor$-node whose unique child is $q_b$. The operation takes time $O(B''')$. We add one $\lor$-node for the source of the nFBDD if needed.
\end{proof}

The nodes of $1$-complete $0$-reduced alternating nFBDD are organized in layers $L_0,L_1,\dots,$ $L_{2n}$. $L_0$ contains the sinks and, for $1 \leq i \leq 2n$ the layer $L_i$ contains all nodes whose children (except the $0$-sink) are in $L_{i-1}$. We write $L_{\leq i} = L_0 \cup \dots \cup L_i$, and similarly for $L_{\geq i}$, $L_{< i}$ and $L_{> i}$. Note that for all $1 \leq i \leq n$, $L_{2i -1 }$ contains only decision nodes whereas $L_{2i}$ contains only $\lor$-nodes. Importantly, we will assume an arbitrary ordering on the children of the nodes; for every node $q$ we have a sequence (not a set) $\children(q)$ of its children.

\paragraph*{FPRAS.}

For a counting problem that, given an input $x$ of size $|x|$, aims at computing some integer value $N(x)$, a fully polynomial-time randomized approximation scheme (FPRAS) is an algorithm that, given $x$, $\varepsilon > 0$, and $0 < \delta < 1$, runs in time polynomial in $|x|$, $1/\varepsilon$, and $\log(1/\delta)$, and returns $\tilde{N}$ with the guarantee that
$
\Pr\left[ \tilde{N} \in (1 \pm \varepsilon)N(x)\right] \geq 1 - \delta.
$
In this paper we give an \#FPRAS for the problem \#nFBDD.\\

\begin{center}
\noindent\fbox{
\parbox{0.9\textwidth}{{\#nFBDD}

$\quad${\bf Input}: an nFBDD $B$

$\quad${\bf Output}: its number of models $|B^{-1}(1)|$
}}
\end{center}

\subsection{Related Work}\label{sec:relatedwork}

As noted in Section~\ref{section:introduction}, the literature on binary decision diagrams is extensive; therefore, we will focus solely on related results in the context of the model counting problem. The problem of \#nFBDD is \#P-complete: membership in \#P is immediate as every assignment can be evaluated in PTIME, and the \#P-hardness follows from the observation that the problem of \#DNF, 
i.e., counting the number of satisfying assignments of Boolean formulas in Disjunctive Normal Form, is \#P-hard~\cite{V79}. Moreover, every DNF can be represented as an nFBDD such that the size of the resulting nFBDD is polynomial in the size of the DNF. Furthermore, it is also known that the problem of \#nOBDD is SpanL-complete~\cite{ACJR19}.

Given the \#P-hardness, a natural direction of research is to understand the complexity of approximation. The discovery of polynomial-time randomized approximation schemes for \#P-hard problems has been of long-standing interest and has led to several theoretically deep and elegant results. One such result was that of Karp and Luby~\cite{KL83} in the context of \#DNF, relying on Monte Carlo sampling. Building on Monte Carlo sampling, Kannan, Sampath, and Mahaney~\cite{KSM95} proposed a quasi-polynomial running time approximation scheme for \#nOBDD.\footnote{The result of ~\cite{KSM95} was stated for regular languages, but the underlying techniques can handle \#nOBDD.} In a follow-up work~\cite{GJKSM97}, this result was extended to handle context-free grammars by reducing the problem of counting words of a context-free grammar to estimating the support size of multilinear $(+,\times)$-programs. It is straightforward to see that the same reduction can be applied in the context of \#DNNF, implying a quasi-polynomial runtime approximation for \#nFBDD. Since then, the question of the existence of a fully polynomial-time randomized approximation scheme for \#nFBDD and its generalizations has remained open.

In a major breakthrough, Arenas et al.~\cite{ACJR19} provided an FPRAS for \#nOBDD. Their technique relied on union of sets estimation \`a la Karp-Luby and the generation of independent samples via the self-reducibility union property.\footnote{The term {\em self-reducibility union property} was coined by Meel, Chakraborty, and Mathur~\cite{MCM24} to explain the high-level idea of~\cite{ACJR19}.} The self-reducibility union property can be informally stated as follows: The set of models conditioned on a given partial assignment (according to the variable ordering of the given nOBDD) can be expressed as the union of models of the states of the given nOBDD.  In a follow-up work~\cite{ACJR21}, Arenas et al. observed that the problem of model counting over structured DNNF (st-DNNF) circuits also admits FPRAS. 
In this context, it is worth highlighting that the self-reducibility union property does not hold for nFBDD and there exists exponential separation between nFBDD and st-DNNF, i.e., there is a family of functions for which the smallest FBDD is exponentially smaller than the smallest st-DNNF, and therefore, the problem of whether there exists an FPRAS for \#nFBDD remains open.

\section{Technical Overview}\label{section:technical_overview}

Our algorithm proceeds in a bottom-up manner and for every node $q$ of given nFBDD $B$, we  keep: (1) a number $p(q) \in (0,1]$ which seeks to approximate $\frac{1}{|\mods(q)|}$, and therefore, $\frac{1}{p(q)}$ can be used to estimate $|\mods(q)|$, and (2) sets of samples $\{S^r(q)\}_{r=1}^{n_s n_t}$ of $\mods(q)$ where $n_s$ and $n_t$ are polynomial in $n = |var(B)|$, $\varepsilon^{-1}$ and $\log|B|$ respectively. Few comments are in order: we keep many  ($n_s \cdot n_t$) independent sets of samples so as to rely on the median of means estimator.  As mentioned earlier,  our algorithm works in bottom-up manner, in particular, for a node $q$, we will compute $(p(q), \{S^r(q)\}_r)$ using the values of $p(\cdot)$ and $\{S^r(\cdot)\}_r$ of its children.

Ideally, we want every model of $q$ to be in $S^r(q)$ identically and independently with probability $p(q)$. However, obtaining iid samples is computationally expensive, which resulted in quasi-polynomial runtimes in earlier studies~\cite{GJKSM97}. Recent works on FPRAS for nOBDD achieved independence by leveraging self-reducibility union property~\cite{ACJR19,ACJR21}, but, as remarked in Section~\ref{sec:relatedwork}, the self-reducibility union property does not hold for nFBDD and therefore, it is not known how to accomplish independence among samples without giving up on the desiderata of polynomial time.

The key insight in our approach is to give up the desire for independence altogether, in particular, we do not even ensure pairwise independence, i.e., even for $\alpha, \alpha' \in \mods(q)$, it may not be the case that $\Pr[\alpha \in S^r(q) | \alpha' \in S^{r}(q)] = \Pr[\alpha \in S^r(q)]$. Of course, we do need to quantify the dependence. In order to discuss the dependence, we first discuss how we update $p(q)$ and $S^r(q)$ for decision nodes and $\lor$-nodes. 
 
\begin{itemize}	
	\item Let $q = ite(x,q_1,q_0)$. Then we compute $p(q)$ and $S^r(q)$ as $p(q) = \left(\frac{1}{p(q_0)}+ \frac{1}{p(q_1)}\right)^{-1}$ and $S^{r}(q) = \left(\reduce\left(S^r(q_0), \frac{p(q)}{p(q_0)}\right) \otimes \{x \mapsto 0\}\right)  \cup \left(\reduce\left(S^r(q_1) , \frac{p(q)}{p(q_1)}\right) \otimes \{x \mapsto 1\}\right)$, where $\reduce(S,p)$ is the operation that keeps each element of a set $S$ with probability $p$.
	
 	\item Let $q$ be a $\lor$-node such that $q = q_1 \lor q_2$ (assuming two children for simplicity). Furthermore, for simplicity of exposition and for the sake of high-level intuition, assume $p(q_1) = p (q_2)$. The technical difficulty for $\lor$-nodes arises from the fact that it is possible that for a given $\alpha \in \mods(q)$, we have  $\alpha \in \mods(q_1) \cap \mods(q_2)$. Therefore, in order to ensure no $\alpha$ has undue advantage even if it lies in the models set of multiple children, we update $p(q)$ and $S^r(q)$ as follows: we first compute $
		\rho = \min(p(q_1),p(q_2))$ and $\hat{S}^r(q) = S^r(q_1) \cup (S^r(q_2) \setminus \mods(q_1))$
		and then $p(q) = \underset{0 \leq j < n_t}{\median}\left(\frac{1}{\rho \cdot n_s}\sum_{r = j\cdot n_s+1}^{(j+1)n_s}|\hat{S}^r(q)|\right)^{-1}$ followed by $S^r(q) = \reduce\left(\hat{S}^r(q), \frac{p(q)}{\rho}\right)$.
\end{itemize}

 Observe that  the usage of  $S^r(q_2) \setminus \mods(q_1)$ ensures that for every $\alpha \in \mods(q)$, there is exactly one child of $q' \in \children(q)$ such that if $\alpha \in \hat{S}^{r}(q)$ then $\alpha \in S^r(q')$, and therefore, no $\alpha$ has {\em undue advantage}. 

It is worth re-emphasizing the crucial departure in our work from  earlier efforts is embrace of dependence. For instance, consider $q = q_1 \lor q_2$ and $\hat{q} = q_1 \lor q_3$, then $S^{r}(q)$ and $S^r(\hat{q})$ will of course be reusing samples $S^r(q_1)$, and therefore, do not even have pairwise independence. Now, of course, we need to bound dependence so as to retain any hope of computing $p(q)$ from $S^{r}(q)$. To this end, we focus on the following quantity of interest: $\Pr[\alpha \in S^{r}(q) \mid \alpha' \in S^{r}(q)]$ for $\alpha, \alpha' \in \mods(q)$, which determines the variance for the estimator. In this regard, for every $(\alpha, q)$, we can define a derivation path, denoted as $\paths(\alpha, q)$, where for every $\lor$-node, $\paths(\alpha,q)$ is $\paths(\alpha,q')$ appended with $q$, where $q'$ is the first child of q such that $\alpha \in \mods(q')$.  The key observation is that our computations ensure $\Pr[\alpha \in S^r(q) \mid \alpha' \in S^r(q)]$ depends on the first node (starting from the $1$-sink) $q^*$ where the derivation paths $\paths(\alpha,q)$ and $\paths(\alpha,q')$ diverge. In particular, it turns out: 
\[
	\Pr[\alpha \in S^r(q) | \alpha' \in S^r(q)] \leq \frac{p(q)}{p(q^*)}.
\]

One might wonder whether the above expression suffices: it turns out it does, because the number of pairs $(\alpha,\alpha')$ whose derivation paths diverge for the first time at $q^*$ can be shown to be bounded by $\frac{|\mods(q)|^2}{|\mods(q^*)|}$, which suffices to show that the variance of the estimator can be bounded by constant factor of square of its mean, thus allowing us to use median of means estimator. 

We close off by remarking that  the situation is more nuanced than previously described, as $p(q)$ is itself a random variable. Although the high-level intuition remains consistent, the technical analysis requires coupling, based on a carefully defined random process, detailed in 
	Section~\ref{section:random_process}. Simplifying the rigorous technical argument would be an interesting direction of future research.

\section{Algorithm}\label{section:algorithm}

The core of our FPRAS, $\approxMCnFBDDcore$, takes in a $1$-complete, $0$-reduced and alternating nFBDD $B$. $B$'s nodes are processed in bottom-up order, so from the sinks to the source. For each node $q$, the algorithm computes $p(q)$ which seeks to estimate $|\mods(q)|^{-1}$, and polynomially-many subsets of $\mods(q)$ called \emph{sample sets}: $S^1(q),\dots,S^{N}(q)$. 
The algorithm stops after processing the source node $q_\text{source}$ and returns $p(q_\text{source})^{-1}$. The procedure that computes the content of $S^r(q)$ and the value for $p(q)$ is $\estimateAndSample(q)$. Since this procedure uses randomness, the $\{S^r(q)\}_r$ and $p(q)$ are random variables. %
Our FPRAS works towards ensuring that ``$\Pr[\alpha \in S^r(q)] = p(q)$'' holds true for every $q \in B$ and $\alpha \in \mods(q)$. Thus, if $p(q)$ is a good estimate of $|\mods(q)|^{-1}$, then $S^r(q)$ should be small in expectation. We put the equality between quotes because it does not make much sense as the left-hand side is a probability, so a fixed number, whereas the right-hand side is a random variable.

\begin{algorithm}
$p(q) = \textsf{round}(q,(\frac{1}{p(q_0)} + \frac{1}{p(q_1)})^{-1})$\label{line:estimateDecision}\\
\label{line:forReduceDecision}
\For{$1 \leq r \leq n_sn_t$}{
	$S^r(q) = \left(\reduce\left(S^r(q_0),\frac{p(q)}{p(q_0)}\right) \otimes \{x \mapsto 0\}\right) \cup \left(\reduce\left(S^r(q_1),\frac{p(q)}{p(q_1)}\right) \otimes \{x \mapsto 1\}\right)$\\\label{line:reduceDecision}
}
\caption{$\estimateAndSample(q)$ with $q = ite(x,q_1,q_0)$}\label{algorithm:estimateAndSampleDecision}
\end{algorithm} 

Decision nodes and $\lor$-nodes are handled differently by $\estimateAndSample$. When $q$ is a decision node $ite(x,q_1,q_0)$, $p(q)$ is computed deterministically from $p(q_0)$ and $p(q_1)$, and $S^r(q)$ is \emph{reduced from} $(S^r(q_0) \otimes \{x \mapsto 0\}) \cup (S^r(q_1) \otimes \{x \mapsto 1\})$ using the $\reduce$ procedure.

\begin{algorithm}
$S' \leftarrow \emptyset$\\
\For{$s \in S$}{
	add $s$ to $S'$ with probability $p$\\
}
\Return{$S'$}
\caption{$\reduce(S,p)$ with $p \in [0,1]$}
\end{algorithm}

$\estimateAndSample(q)$ is more complicated when $q$ is a $\lor$-node. We explain it gradually. For starter, let $q = q_1 \lor \dots \lor q_k$ with its children ordered as follows: $(q_1,\dots,q_k)$, and consider the problem of approximating $|\mods(q)|$ when a sample set $S(q_i) \subseteq \mods(q_i)$ is available for every $i \in [k]$ and $b \in \{0,1\}$ with the guarantee that $\Pr[\alpha \in S(q_i)] = \rho(q)$ holds for every $\alpha \in \mods(q_i)$. Every $S(q_i)$ is a subset of $\mods(q)$. We compute $\hat S(q)  = \union(q,S(q_1),\dots,S(q_k))$ as follows: for every $\alpha \in S(q_i)$, $\alpha$ is added to $\hat S(q)$ if and only if $q_i$ is the first child of $q$ for which $\alpha$ is a model. Simple computations show that $\Pr[\alpha \in \hat S(q)] = \rho(q)$ holds for every $\alpha \in \mods(q)$, and therefore $\rho(q)^{-1}|\hat S(q)|$ is an unbiased estimate of $|\mods(q)|$ (i.e., the expected value of the estimate is $|\mods(q)|$).

\begin{algorithm}
$S' = \emptyset$\\
\For{$1 \leq i \leq k$}{
	\For{$\alpha \in S_i$}{
		\textbf{if} $\alpha \not\in \mods(q_j)$ for every $j < i$ \textbf{then} add $\alpha$ to $S'$\\ 
	}
}
\Return{$S'$}
\caption{\mbox{$\union(q,S_1,\dots,S_k)$ with $\children(q) = (q_1,\dots,q_k)$}}\label{algorithm:union}
\end{algorithm}

Now suppose that for every $i$, we only know that $\Pr[\alpha \in S(q_i)] = p(q_i)$ for every $\alpha \in \mods(q_i)$ for some number $p(q_i)$ independent of $\alpha$. Then we normalize the probabilities before computing the $\union$. This is done using the $\reduce$ procedure. Let $\rho(q) = \min(p(q_1),\dots,p(q_k))$ and $\bar S(q_i) = \reduce(S(q_i), \frac{\rho(q)}{p(q_i)})$. We have that $\Pr[\alpha \in \bar S(q_i)] = \rho(q)$ holds for every $\alpha  \in \mods(q)$. So now $\hat S(q) = \union(q,\bar S(q_1),\dots,\bar S(q_k))$, and $\rho(q)^{-1}|\hat S(q)|$ is an unbiased estimate of $|\mods(q)|$.
\begin{algorithm}
$\rho(q) = \textsf{min}(p(q_1),\dots,p(q_k))$\label{line:rho}\\
\For{$1 \leq r \leq n_s n_t$}{
	$\hat{S}^r(q) = \union\left(q,\reduce\big(S^r(q_1),\frac{\rho(q)}{p(q_1)}\big),\dots,
	\reduce\big(S^r(q_k),\frac{\rho(q)}{p(q_k)}\big)\right)$\label{line:union}\\
}
\label{line:forMean}
\For{$1 \leq j \leq n_t$}{
$M_j(q) = \frac{1}{\rho(q) n_s}\sum_{r = (j-1) n_s+1}^{j\cdot n_s}|\hat{S}^r(q)|$\label{line:mean}\\
}
$\hat\rho(q) = \median(M_1(q),\dots,M_{n_t}(q))^{-1}$ \label{line:median}\\
$p(q) = \textsf{round}(q,\textsf{min}(\rho(q),\hat\rho(q)))$\label{line:round}\\ \label{line:forReduce}
\For{$1 \leq r \leq n_sn_t$}{
 	$S^r(q) = \reduce\big(\hat{S}^r(q),\frac{p(q)}{\rho(q)}\big)$\label{line:reduce}\\
}
\caption{$\estimateAndSample(q)$ with $q = q_1 \lor \dots \lor q_k$ and $\children(q) = (q_1,\dots,q_k)$}\label{algorithm:estimateAndSampleOr}
\end{algorithm}

To find an estimate that is concentrated around $|\mods(q)|$, we use the ``median of means'' technique. Suppose that instead of one sample set $S(q_i)$ we have several sample sets $S^1(q_i),S^2(q_i),\dots,S^N(q_i)$ all verifying $\Pr[\alpha \in S^r(q_i)] = p(q_i)$. Then define $\bar S^r(q_i)$ similarly to $\bar S(q_i)$ and $\hat S^r(q)$ similarly to $\hat S(q)$. Each $\rho(q)^{-1}|\hat S^r(q)|$ is an estimate of $|\mods(q)|$. Say $N = n_sn_t$ and partition $\hat S^1(q),\hat S^2(q),\dots,\hat S^N(q)$ into $n_t$ batches of $n_s$ sets. The median of means technique computes the average normalized size $M_j(q) = \frac{1}{\rho(q)n_s}\sum_{r = (j - 1)n_s + 1}^{j\cdot n_s} |\hat S^r(q)|$ over each batch and uses $\median(M_1(q),\dots,M_{n_t}(q))$ to estimate $|\mods(q)|$. The mean computation aims at reducing the variance of the estimate. The parameter $n_s$ can be chosen so that $\Pr[M_j(q) \in (1 \pm \varepsilon)|\mods(q)|] > \frac{1}{2}$ holds. With the appropriate value for $n_t$, $\median(M_1(q),\dots,M_{n_t}(q))$ lies in $(1 \pm \varepsilon)|\mods(q)|$ with high probability (though the estimate is not unbiased anymore).

So this is how $\estimateAndSample$ works for $\lor$-nodes. The median of means serves to compute $\hat \rho(q)$, the inverse of the estimate of $|\mods(q)|$ which in turn is used to compute $p(q)$. When $q$ is not the source node $q_\text{source}$, we compute sample sets $S^r(q)$ in preparation for processing of $q$'s ancestors. For this we reuse $\hat{S}^r(q)$ and compute $S^r(q) = \reduce(\hat{S}^r(q),\frac{p(q)}{\rho(q)})$. 

\begin{algorithm}
$p(1\text{-sink}) = 1$, $p(0\text{-sink}) = \infty$\\
\For{$1 \leq r \leq n_sn_t$}{
	$S^r(1\text{-sink}) = \{\alpha_\emptyset\}$, $S^r(0\text{-sink}) = \emptyset$\\
}
\For{$1 \leq i \leq 2n$}{
	\For{$q \in L_i$}{
		$\estimateAndSample(q)$\\ 
		\textbf{if} $|S^r(q)| \geq \theta$ \textbf{then return} $0$\label{line:interrupt} \\
	}
}
\Return{$1/p(q_{\text{source}})$}
\caption{$\approxMCnFBDDcore(B,n,n_s,n_t,\theta)$}
\end{algorithm}

To ensure a polynomial running time, $\approxMCnFBDDcore$ terminates as soon as the number of samples grows too large (Line~\ref{line:interrupt}) and returns $0$. This output is erroneous but we will show that the probability of terminating this way is negligible. For parameters carefully chosen, $\approxMCnFBDDcore$ returns a good estimate of $|B^{-1}(1)|$ with probability larger than $1/2$. The full FPRAS $\approxMCnFBDD$ amplifies this probability to $1 - \delta$ by returning the median output of independent runs of $\approxMCnFBDD$.

\begin{algorithm}
make $B$ alternating, $1$-complete and $0$-reduced\\
$n = |var(B)|$, $m = \lceil 8\ln(1/\delta)\rceil$,
$\kappa = \varepsilon/(1 + \varepsilon)$, 
$n_s = \lceil 4n/\kappa^2 \rceil$, $n_t = 16n|B|$,  
$\theta = \lceil 16n_sn_t(1+\kappa)|B| \rceil$\\
\For{$1 \leq j \leq m$}{
	$\mathsf{est}_j = \approxMCnFBDDcore(B,n,n_s,n_t,\theta)$\\
}
\Return $\median(\mathsf{est}_1,\dots,\mathsf{est}_m)$ 
\caption{$\approxMCnFBDD(B,\varepsilon,\delta)$}
\end{algorithm}

\section{Derivation paths}\label{section:derivation_path} 

Models can have several accepting paths in an nFBDD. For $q$ a node of $B$ and $\alpha \in \mods(q)$, we map $(\alpha,q)$ to a canonical accepting path, called the \emph{derivation path} of $\alpha$ for $q$, denoted by $\paths(\alpha,q)$. A path $\calP$ is formally represented with a tuple $(V(\calP),E(\calP))$, with $V(\calP)$ a sequence of vertices and $E(\calP)$ a sequence of edges.

\begin{definition}
Let $q \in B$ and $\alpha \in \mods(q)$. The derivation path $\paths(\alpha,q)$ is defined as follows:
\begin{itemize}
\item If $q$ is the $1$-sink then $\alpha = \alpha_\emptyset$ and the only derivation path is $\paths(\alpha_\emptyset,q) = (\{q\},\emptyset)$.
\item If $q = ite(x,q_1,q_0)$, let $\alpha'$ be the restriction of $\alpha$ to $var(\alpha) \setminus \{x\}$, then $V(\paths(\alpha,q)) = V(\paths(\alpha',q_{\alpha(x)})) \cdot q$ and $E(\paths(\alpha,q))$ is $E(\paths(\alpha',q_{\alpha(x)}))$ plus the $\alpha(x)$-edge of $q$.
\item If $q = q_1 \lor \dots \lor q_k$ with the children ordering $\children(q) = (q_1,\dots,q_k)$, let $i$ be the smallest integer between $1$ and $k$ such that $\alpha \in \mods(q_i)$ then $V(\paths(\alpha,q)) = V(\paths(\alpha,q_i)) \cdot q$ and $E(\paths(\alpha,q))$ is $E(\paths(\alpha,q_i))$ plus the edge between $q_i$ and $q$.
\end{itemize} 
\end{definition}

Our algorithm constructs sample sets in a way that respect derivation paths. That is, an assignment $\alpha \in \mods(q)$ may end up in $S^r(q)$ only if it is derived through $\paths(\alpha,q)$.

\begin{restatable}{lemma}{algoConsistentWithDerivationPath}
Let $q \in B$ and $\alpha \in \mods(q)$, let $V(\paths(\alpha,q)) = (q_0,q_1,\dots,q_{i-1},q_i)$ with $q_0 $ the $1$-sink and $q_i = q$. For every $j \in [0,i-1]$, let $\alpha_j$ be the restriction of $\alpha$ to $var(q_j)$. In a run of $\approxMCnFBDDcore$, $\alpha \in S^r(q)$ holds only if $\alpha_j \in S^r(q_j)$ holds for every $j \in [0,i-1]$.
\end{restatable}
\begin{proof}
$\alpha_0 = \alpha_\emptyset \in S^r(1\text{-sink}) = S^r(q_0)$ holds by construction. Now consider $j > 0$, it is sufficient to show that $\alpha_j \in S^r(q_j)$ only if $\alpha_{j-1} \in S^r(q_{j-1})$.

If $q_j = ite(x,q_{j,1},q_{j,0})$ then $\alpha_j = \alpha_{j - 1} \cup \{x \mapsto \alpha(x)\}$ and $q_{j-1} = q_{j,\alpha(x)}$. Looking at $\estimateAndSample$ for decision nodes, one sees that $\alpha_j \in S^r(q_j)$ only if $\alpha_{j - 1} \in \reduce(S^r(q_{j,\alpha(x)}), \frac{p(q_j)}{p(q_{j-1})})$, so only if $\alpha_{j-1} \in S^r(q_{j,\alpha(x)}) = S^r(q_{j-1})$. 

If $q_j$ is a $\lor$-node with $\children(q_j) = (q^0_j,\dots,q^k_j)$ then $\alpha_j = \alpha_{j-1}$ and there is an $i$ such that $q_{j-1} = q^i_j$. Looking at $\estimateAndSample$ for decision node, one sees that $\alpha_j \in S^r(q_j)$ only if $\alpha_j \in \hat{S}^r(q_j)$, only if $\alpha_j$ is $\reduce(S^r(q^\ell_j),\rho/p(q^\ell_j))$ for the smallest $\ell$ such that $\alpha_j \in \mods(q^\ell_j)$., so only if $\alpha_j = \alpha_{j-1} \in S^r(q^\ell_j)$. The definition of $\paths(\alpha,q)$ implies that $\ell = i$, so we are done.
\end{proof}

Given two derivation paths $\calP$ and $\calP'$. We call their \emph{last common prefix nodes} denoted by $\lcpn(\calP,\calP')$, the deepest node where the two paths diverge, that is, the first node contained in both paths from which the they follow different edges. Note that if $\calP$ and $\calP'$ are consistent up to node $q'$, and $q = ite(x,q',q')$, and $\calP$ follows the $0$-edge while $\calP'$ follows the $1$-edge, then the two paths diverge at $q'$ even though they both contain $q$. 

\begin{definition} Let $\calP = ((q_0,\dots,q_k),(e_1,\dots,e_k))$ and $\calP' = ((q'_0,\dots,q'_\ell),(e'_1,\dots,e'_\ell))$ be two derivation paths. The last common prefix node, denoted by $\lcpn(\calP,\calP')$, is the node $q_i$ for the biggest $i$ such that $(q_0,\dots,q_i) = (q'_0,\dots,q'_i)$ and $(e_1,\dots,e_i) = (e'_1,\dots,e'_i)$. 
\end{definition}

Note that every derivation path contains the $1$-sink for first node, so the last common prefix node is well-defined. Let $V(\paths(\alpha,q)) = (q_0,\dots,q_i)$ with $q_0$ the $1$-sink and $q_i = q$. For every $0 \leq \ell \leq i$, we define
\[
I(\alpha,q,\ell) := \{\alpha' \in \mods(q) \mid \lcpn(\paths(\alpha,q),\paths(\alpha',q)) = q_\ell\}. 
\]
The following result will play a key role in bounding the variance of estimators in the analysis.

\begin{restatable}{lemma}{derivationPath}\label{lemma:derivation_path}
Let $\alpha \in \mods(q)$ and $V(\paths(\alpha,q)) = (q_0,\dots,q_i)$ with $q_0$ the $1$-sink and $q_i = q$.  For every $0 \leq \ell \leq i$, $|I(\alpha,q,\ell)| \leq \frac{|\mods(q)|}{|\mods(q_\ell)|}$.
\end{restatable}
\begin{proof}
Let $I(\alpha,q,\ell) = \{\alpha^1,\alpha^2,\dots\}$. Let $\alpha_\ell$ be the restriction of $\alpha$ to $var(q_\ell)$. By definition, every $\alpha^i$ is of the form $\alpha_\ell \cup \beta^i$ for some assignment $\beta^i$ to $var(q) \setminus var(q_\ell)$. Consider $\alpha'_\ell \in \mods(q_\ell)$, then every $\alpha'_\ell \cup \beta^i$ is in $\mods(q)$. Since the assignments $\alpha^i$s differ on $var(q) \setminus var(q_\ell)$, the $\beta^i$s are pairwise distinct, and therefore $\{\alpha'_\ell \cup \beta^i\}_i$ is a set of $|I(\alpha,q,\ell)|$ distinct assignments in $\mods(q)$. Considering all $|\mods(q_\ell)|$ possible choices of assignment for $\alpha'_\ell$, we find that $\{\alpha'_\ell \cup \beta^i\}_{\alpha'_\ell,i}$ is a set of $|\mods(q_\ell)|\cdot|I(\alpha,q,\ell)|$ distinct assignments in $\mods(q)$. Hence $|\mods(q_\ell)|\cdot|I(\alpha,q,\ell)| \leq |\mods(q)|$.
\end{proof}
\section{The Framework for the Analysis}\label{section:random_process}

We introduce a random process that simulates $\approxMCnFBDDcore$. Our intuition is that, for every $\alpha \in \mods(q)$, a statement in the veins of ``$\Pr[\alpha \in S^r(q)] = p(q)$'' should hold. The problem is that this equality makes no sense because $\Pr[\alpha \in S^r(q)]$ is a fixed real value whereas $p(q)$ is a random variable. The variables $S^r(q)$ and $p(q)$ for different $q$ are too dependent of each other so we use a random process to work with new variables that behave more nicely. The global picture is that the random process simulates all possible runs of the algorithm for all possible values of $(p(q))_q$ at once. In the random process, $S^r(q)$ is simulated by a different variable for each possible run. We can then ``replace'' $S^r(q)$ by one of these variables assuming enough knowledge on the algorithm run up to $q$. This knowledge is recorded in what we call a \emph{history} for $q$.

\subsection{History}

A \emph{history} $h$ for a set of nodes $Q$ is a mapping
$
h : Q \rightarrow \mathbbm{Q} \cup \{\infty\}.
$
$h$ is \emph{realizable} when there exists a run of $\approxMCnFBDDcore^*$ that gives the value $h(q)$ to $p(q)$ for every $q \in Q$. Such a run is said \emph{compatible} with $h$. Two histories $h$ for $Q$ and $h'$ for $Q'$ are \emph{compatible} when $h(q) = h'(q)$ for all $q \in Q \cap Q'$. Compatible histories can be merged into an history $h \cup h'$ for $Q \cup Q'$. For $q \in Q$ and $t \in \mathbbm{Q}$, we write $h \cup (q \mapsto t)$ to refer to the history $h$ augmented with $h(q) = t$. For $q \in B$, we define the set $\descendants(q)$ of its descendants by $\descendants(1\text{-sink}) = \emptyset$ and $\descendants(q) = \children(q) \cup \bigcup_{q' \in \children(q)} \descendants(q')$. Note that $q \not\in \descendants(q)$. We only study histories realizable for sets $Q$ that are closed for $\descendants$, that is, if $q \in Q$ and $q'$ is a descendant of $q$, then $q' \in Q$. Thus we abuse terminology and refer to a history for $\descendants(q)$ as a history for $q$. The only history for the sinks is the vacuous history $h_\emptyset$ for $Q = \emptyset$ (because no descendants).

\subsection{Random Process} 

The random process comprises $n_sn_t$ independent copies identified by the superscript $r$. For $q$ a node of the nFBDD, $t \in \mathbbm{Q} \cup \{\infty\}$ and $h$ a realizable history for $q$, we have a random variable 
$\vY{r}{h}{t}{q}$ whose domain is all possible subsets of $\mods(q)$. $\vY{r}{h}{t}{q}$ is meant to simulate $S^r(q)$ when the history for $q$ is $h$ and the value $t$ is assigned to $p(q)$. The variables $\vY{r}{h}{t}{q}$ are defined for realistic values of $t$ (i.e., values that can be given to $p(q)$ by the algorithm) in an inductive way.
\begin{itemize}
\item If $q$ is the $0$-sink, then $\mods(q) = \emptyset$ and only $\vY{r}{h_\emptyset}{\infty}{q} = \emptyset$ is defined.
\item If $q$ is the $1$-sink, then $\mods(q) = \{\alpha_\emptyset\}$ and only $\vY{r}{h_\emptyset}{1}{q} = \{\alpha_\emptyset\}$ is defined.
\item If $q = ite(x,q_1,q_0)$ then for every $\vY{r}{h_0}{t_0}{q_0}$ and $\vY{r}{h_1}{t_1}{q_1}$ with $h_1$ and $h_0$ realizable and compatible histories for $q_1$ and $q_0$, respectively, we define $h = h_1 \cup h_0 \cup (q_0 \mapsto t_0, q_1 \mapsto t_1)$ and $t = \left(\frac{1}{t_0} + \frac{1}{t_1}\right)^{-1}$ and
\[
\vY{r}{h}{t}{q} = \reduce\left(\vY{r}{h_0}{t_0}{q_0} \otimes \{x \mapsto 0\}, \frac{t}{t_0}\right) \cup \reduce\left(\vY{r}{h_1}{t_1}{q_1} \otimes \{x \mapsto 1\}, \frac{t}{t_1}\right).
\] 
\item If $q = q_1 \lor \dots \lor q_k$ then for every $\vY{r}{h_1}{t_1}{q_1}$, $\dots$, $\vY{r}{h_k}{t_k}{h_k}$ with realizable and pairwise compatible histories, we define $h = h_1 \cup \dots \cup h_k \cup (q_1 \mapsto t_1,\dots,q_k \mapsto t_k)$, $t_{\min} = \min(t_1,\dots,t_k)$ and the variable $\vZ{r}{h}{q}$ that simulates $\hat{S}^r(q)$ when the history for $q$ is $h$
\[
\vZ{r}{h}{q} = \union\left(q,\reduce\left(\vY{r}{h_1}{t_1}{q_1},\frac{t_{\min}}{t_1}\right), \dots ,\reduce\left(\vY{r}{h_k}{t_k}{h_k},\frac{t_{\min}}{t_k}\right)\right).
\]
Now for all $t \leq t_{\min}$ we define
$
\vY{r}{h}{t}{q} = \reduce\left(\vZ{r}{h}{q},\frac{t}{t_{\min}}\right).
$
\end{itemize}

Let $H(q)$ be the random variable on the history for $q$ obtained when running $\approxMCnFBDDcore^*$. The following lemma is essentially saying that a probability $\Pr[H(q) = h \text{ and }event]$ where $event$ is an event over random variables defined by $\approxMCnFBDDcore^*$ for $q$ and descendants of $q$ and can be replaced by $\Pr[event_h]$ where $event_h$ is the same event as $event$ but where the variables are replaced by their counterpart in the random process for the history $h$.

\begin{restatable}{lemma}{notBlackMagic}\label{lemma:not_so_black_magic}
Let $\vec S(q) = (S^r(q) \mid r \in [n_sn_t])$, $\vec Z(q) = (\hat S^r(q) \mid r \in [n_sn_t])$, $\vec \calS(q) = Z(q) \cup \bigcup_{q' \in \descendants(q)} \vec S(q') \cup \vec Z(q')$, $\vec{p}(q) = (p(q') \mid q' \in \descendants(q))$ and $e(\vec \calS(q), \vec p(q))$ be an event function of $\vec \calS(q)$ and $\vec p(q)$. For $h$ an history for $q$, let $\vec{h}(q) = (h(q') \mid q' \in \descendants(q))$ and let $\vec{\mathfrak{S}}_h(q)$ be the same sequence as $\vec \calS(q)$ where each $S^r(q')$ is replaced by $\mathfrak{S}^r_{h_{q'},h(q')}(q')$ with $h_{q'}$ the restriction of $h$ to $\descendants(q')$ and where each $\hat S^r(q')$ is replaced by $\hat {\mathfrak{S}}^r_{h_{q'}}(q')$. Then
$$
\Pr[H(q) = h \text{ and } e(\vec \calS(q), \vec p(q))] \leq \Pr[e(\vec{\mathfrak{S}}_h(q), \vec{h}(q)]
$$
\end{restatable}

\noindent The lemma above and the two next lemmas are proved in appendix. These lemmas give the sound variants of ``$\Pr[\alpha \in S^r(q)] = p(q)$'' in the random process.

\begin{restatable}{lemma}{probaFirstOrder}\label{lemma:proba_first_order}
For every $\vY{r}{h}{t}{q}$ and $\alpha \in \mods(q)$, it holds that
$\Pr[\alpha \in \vY{r}{h}{t}{q}] = t$.
In addition, if $q$ is a $\lor$-node with $\children(q) = (q_1,\dots,q_k)$ then $\Pr[\alpha \in \vZ{r}{h}{q}] = \min(h(q_1),\dots,h(q_k))$. 
\end{restatable}

\begin{restatable}{lemma}{probaSecondOrder}\label{lemma:proba_second_order}
For every $\vY{r}{h}{t}{q}$ with $\children(q) = (q_1,\dots,q_k)$ and $\alpha,\alpha' \in \mods(q)$ with $\alpha \neq \alpha'$, let $t^* = h(\lcpn(\paths(\alpha,q),\paths(\alpha',q))$, then we have that
$
\Pr[\alpha \in \vY{r}{h}{t}{q} \mid \alpha' \in \vY{r}{h}{t}{q}] \leq \frac{t}{t^*}
$ 
and if $q$ is $\lor$-node then
$
\Pr[\alpha \in \vZ{r}{h}{q} \mid \alpha' \in \vZ{r}{h}{q}] \leq \frac{\min(h(q_1),\dots,h(q_k))}{t^*}.
$
\end{restatable}

\section{Analysis}\label{section:analysis}
We now conduct the analysis of \approxMCnFBDD. The hardest part to analyze is the core algorithm \approxMCnFBDDcore, for which we will prove the following.

\begin{lemma}\label{lemma:main_result_core}
	Let $B$ be a $1$-complete $0$-reduced and alternating nFBDD over $n$ variables. Let $m = \max_i |L^i|$, $\varepsilon > 0$, and $\kappa = \frac{\varepsilon}{1 + \varepsilon}$. If $n_s = \lceil \frac{4n}{\kappa^2}\rceil $, $n_t = 16 n|B|$ and $\theta = \lceil 16n_sn_t(1+\kappa)|B| \rceil$ then $\approxMCnFBDDcore(B,n,n_s,n_t,\theta)$ runs in time $O(n_sn_t \theta |B|^2)$ and returns $\mathsf{est}$ with the guarantee
$
		\Pr\left[\mathsf{est}  \notin (1 \pm \varepsilon) |B^{-1}(1)| \right] \leq \frac{1}{4}.
$
\end{lemma}
The probability $1/4$ is decreased down to any $\delta > 0$ with the median technique. Thus giving our main result. 
\mainResult*
\begin{proof}
Let $\mathsf{est}_1,\dots,\mathsf{est}_m$ be the estimates from $m = O(\log(\delta^{-1}))$ independent calls to $\approxMCnFBDDcore$. Let $X_i$ be the indicator variable that takes value $1$ if and only if $\mathsf{est}_i \not\in (1\pm \varepsilon)|B^{-1}(1)|$, and define $\bar X = X_1 + \dots + X_m$. By Lemma~\ref{lemma:main_result_core}, $\Ex[\bar X] \leq m/4$. Hoeffding bound gives
$$
\Pr\left[\median(X_1,\dots,X_m) \not\in (1 \pm \varepsilon)|B^{-1}(1)|\right] = \Pr\left[\bar X > \frac{m}{2}\right] \leq \Pr\left[\bar X - \Ex[\bar X] > \frac{m}{4}\right] \leq e^{-m/8} \leq \delta.
$$
The running time is $O(m)$ times that of $\approxMCnFBDDcore(B',n,n_s,n_t)$ where $B'$ is $B$ after it has been made $1$-complete, $0$-reduced, and alternating. By Lemma~\ref{lemma:makeComplete}, $|B'| = O(n|B|^2)$ and $B'$ is constructed in time $O(n|B|^2)$. So each call to $\approxMCnFBDDcore(B',n,$ $n_s,n_t)$ takes time $O(n^{11}|B|^{10}\varepsilon^{-4})$
\end{proof}

Recall that $\approxMCnFBDDcore^*$ is $\approxMCnFBDDcore$ without the terminating condition of Line~\ref{line:interrupt}, that is, the same algorithm but where any $|S^r(q)|$ can grow big. Analyzing $\approxMCnFBDDcore^*$ is enough to prove Lemma~\ref{lemma:main_result_core} without running time requirements. In particular, it is enough to prove Lemmas~\ref{lemma:proba_p(q)} and~\ref{lemma:proba_S(q)}. For these lemmas the settings described in Lemma~\ref{lemma:main_result_core} are assumed. 

\begin{lemma}\label{lemma:proba_p(q)}
The probability that $\approxMCnFBDDcore^*(B,n,n_s,n_t,\theta)$ computes $p(q) \not\in (1 \pm \kappa)|\mods(q)|^{-1}$ for some $q \in B$ is at most $1/16$.
\end{lemma}

\begin{lemma}\label{lemma:proba_S(q)}
The probability that $\approxMCnFBDDcore^*(B,n,n_s,n_t,\theta)$ constructs sets $S^r(q)$ such that $|S^r(q)| \geq \theta$ for some $q \in B$ is at most $1/8$.
\end{lemma}

\noindent The combination of Lemmas~\ref{lemma:proba_p(q)} and~\ref{lemma:proba_S(q)} yields Lemma~\ref{lemma:main_result_core}. 

\begin{proof}[Proof of Lemma~\ref{lemma:main_result_core}]
Let $\texttt{A}^{(\ast)} = \approxMCnFBDDcore^{(\ast)}(B,n,n_s,n_t,\theta)$.
\[
\begin{aligned}
\Pr_{\texttt{A}}\left[\mathsf{est} \not\in (1 \pm \varepsilon)|B^{-1}(1)|\right] 
&= \Pr_{\texttt{A}^*}\left[
	\bigcup_{r,q} |S^r(q)| \geq \theta
\right] + \Pr_{\texttt{A}^*}\left[
	\bigcap_{r,q} |S^r(q)| < \theta \text{ and } \mathsf{est} \not\in (1 \pm \varepsilon)|B^{-1}(1)|
\right] 
\\
&\leq \Pr_{\texttt{A}^*}\left[
	\bigcup_{r,q} |S^r(q)| \geq \theta
\right] + \Pr_{\texttt{A}^*}\left[
		\mathsf{est} \not\in (1 \pm \varepsilon)|B^{-1}(1)|
\right]  
\\
&\leq \Pr_{\texttt{A}^*}\left[
	\bigcup_{r,q} |S^r(q)| \geq \theta
\right] + \Pr_{\texttt{A}^*}\left[
	\bigcup_{q}
			p(q) \not\in \frac{1}{(1 \pm \varepsilon)|\mods(q)|} 
\right]
\end{aligned}
\]
where $q$ ranges over $B$'s nodes and $r$ ranges in $[n_sn_t]$. The parameter $\kappa$ has been set so that $p(q) \not\in \frac{1 \pm \kappa}{|\mods(q)|}$ implies $p(q) \not\in \frac{1}{(1 \pm \varepsilon)|\mods(q)|}$ so, using Lemmas~\ref{lemma:proba_p(q)} and~\ref{lemma:proba_S(q)}: 
\[
\Pr_{\texttt{A}}\left[\mathsf{est} \not\in (1 \pm \varepsilon)|B^{-1}(1)|\right] 
\leq \Pr_{\texttt{A}^*}\left[
	\bigcup_{r,q} |S^r(q)| \geq \theta
\right] + \Pr_{\texttt{A}^*}\left[
	\bigcup_{q}
			p(q) \not\in \frac{1 \pm \kappa}{|\mods(q)|} 
\right] \leq \frac{1}{4} 
\]
The algorithm stops whenever the number of samples grows beyond $\theta$ so for the worst-case running time the number of samples is less than $\theta$. Each node goes once through $\estimateAndSample$. For a decision node, $\estimateAndSample$ takes time $O(n_sn_t\theta)$. For a $\lor$-node $q$, $\estimateAndSample$ calls $\union$ $n_sn_t$ times, does a median of means where it computes the median of $n_t$ means of $n_s$ integers, and updates the sample sets. Updating the sample sets takes time $O(n_sn_t\theta)$. Computing each means takes $O(n_s)$ elementary operations and computing the median takes $O(n_t)$ elementary operations by a median of medians procedure. For each sample, the $\union$ tests whether it is a model of the children of $q$, model checking is a linear-time operation on nFBDD so the total cost of one $\union$ is $O(|\children(q)|\cdot|B| \cdot \theta)$. So the total cost of $\estimateAndSample$ for all $\lor$-nodes is at most $O(n_sn_t \theta |B| \cdot \sum_q|\children(q)|) = O(n_sn_t \theta  |B|^2)$ because $|B|$ counts the edges of $B$.
\end{proof}

It remains to prove Lemmas~\ref{lemma:proba_p(q)} and~\ref{lemma:proba_S(q)}.

\subsection{Proof of Lemma~\ref{lemma:proba_p(q)}}

Let $\Delta(q)$ be the interval $\frac{1 \pm \kappa}{|\mods(q)|}$ and $\nabla(q)$ be the interval $\frac{|\mods(q)|}{1 \pm \kappa}$. We want all $p(q)$ to be in $\Delta(q)$. Since the algorithm rounds values when computing $p(q)$, we first explain that we cannot leave $\Delta(q)$ because rounding. $\round(q,\infty)$ returns $\infty$. For $v > 0$, $\round(q,v)$ returns the smallest value greater or equal to $v$ that is of the form $\frac{1 + \kappa}{l}$ for some integer $l$ between $1$ and $2^n$. $\round(q,v)$ returns $\infty$ if no such value exists.

\begin{lemma}\label{lemma:rounding_is_again}
Let $q$ be a node in the nFBDD if $v \in \Delta(q)$ then $\round(q,v) \in \Delta(q)$.
\end{lemma}
\begin{proof}
By convention $\Delta(q) = \{\infty\}$ when $\mods(q) = \emptyset$ so the lemma holds for $v = \infty$ and every node $q$ from which there is no path reaching the $1$-sink. We now consider the other nodes $q$ and thus have $|\mods(q)| \geq 1$. Hence $1 \leq |\mods(q)| \leq 2^n$. Since $\round(q,v)$ is chosen to be minimal, we have that $v  \leq \frac{1 + \kappa}{|\mods(q)|}$ implies $v \leq \round(q,v) \leq \frac{1 + \kappa}{|\mods(q)|}$. Thus $v \in \Delta(q)$ implies $\round(q,v) \in \Delta(q)$.
\end{proof}

\begin{lemma}\label{lemma:there_is_a_fist_bad_event}
The event 
$
		\bigcup_{q \in B}
			\left(p(q) \not\in \Delta(q)
\right)
$
occurs if and only if the following event occurs: 
\[
		\bigcup_{q \in L_{> 0}}\left(
			p(q) \not\in \Delta(q) \text{ and for all } q' \in \descendants(q),\, p(q') \in \Delta(q')
\right)
\]
\end{lemma}
\begin{proof}
The ``if'' direction is trivial. For the other direction, suppose that $p(q) \not\in \Delta(q)$ holds for some $q$. Let $i$ be the smallest integer such that there is $q \in L_i$ and $p(q) \not\in \Delta(q)$. $i$ cannot be $0$ because the only node in $L_0$ is the $1$-sink and $p(1\text{-sink}) = 1 = |\mods(1\text{-sink})|^{-1} \in \Delta(1\text{-sink})$. So $q \in L_{> 0}$ and, by minimality of $i$, we have that $p(q') \in \Delta(q')$ for all $q' \in \descendants(q)$.
\end{proof}
\noindent By Lemma~\ref{lemma:there_is_a_fist_bad_event}
\begin{equation*}
\begin{aligned}
\Pr\bigg[
		\bigcup_{q \in B}
			p(q) \not\in \Delta(q)
\bigg] =
&\Pr\bigg[
		\bigcup_{q \in L_{> 0}}
			p(q) \not\in \Delta(q) 
			\text{ and } \forall q' \in
			 \descendants(q),\, p(q') \in
			  \Delta(q')
\bigg]
\\
\leq &
\sum_{q \in L_{> 0}}\underbrace{\Pr\left[
			p(q) \not\in \Delta(q) 
			\text{ and } \forall q' \in \descendants(q),\, p(q') \in
			  \Delta(q')
\right]}_{P(q)}.
\end{aligned}
\end{equation*}

We bound $P(q)$ from above. The case when $q = ite(x,q_1,q_0)$ is a decision node is easy by construction; if $p(q_0) \in \Delta(q_0)$ and $p(q_1) \in \Delta(q_1)$ then $\frac{1}{p(q_0)}+ \frac{1}{p(q_1)} \in \frac{|\mods(q_0)|+|\mods(q_1)|}{1\pm \kappa} = \frac{|\mods(q)|}{1 \pm \kappa} = \nabla(q)$ with probability $1$ so, by Lemma~\ref{lemma:rounding_is_again}, $p(q) = \round(q,(\frac{1}{p(q_0)}+ \frac{1}{p(q_1)})^{-1}) \in \Delta(q)$ with probability $1$. Thus $P(q) = 0$ when $q$ is a decision node and only the case of $\lor$-nodes remains.

\subsubsection*{Moving to the random process}

Consider the set $\calH_q$ of realizable histories for $q$ and denote by $H(q) = h$ the event that the history $h$ occurs for $q$, that is, the event that the algorithm sets $p(q')$ to $h(q')$ for all $q' \in \descendants(q)$.
\begin{align*}
P(q) &= \Pr\bigg[
			\bigcup_{h \in \calH_q} H(q)=h \text{ and }
				p(q) \not\in \Delta(q) 
				\text{ and for all } q' \in
				 \descendants(q),\, p(q') \in
				  \Delta(q')
		\bigg]
\\
	&\leq \sum_{h \in \calH_q} \Pr\left[H(q)=h \text{ and }
				p(q) \not\in \Delta(q) 
				\text{ and for all } q' \in
				 \descendants(q),\, p(q') \in
				  \Delta(q')
		\right]
\end{align*}
If $H(q) = h$ then the event $p(q') = h(q')$ holds for all $q' \in \descendants(q)$. So the summand probabilities are zero when $h(q')$ is not in $\Delta(q')$ for any descendant of $q$. Let $\calH^*_q$ be the subset of $\calH_q$ where $h(q') \in \Delta(q')$ holds for every $q' \in \descendants(q)$. Then
\begin{equation*}
\begin{aligned}
P(q) &\leq \sum_{h \in \calH^*_q} \Pr\left[H(q)=h \text{ and }
				p(q) \not\in \Delta(q) 
				\text{ and for all } q' \in
				 \descendants(q),\, p(q') \in
				  \Delta(q')
		\right]
\\
	&= \sum_{h \in \calH^*_q} \Pr\left[H(q)=h \text{ and }
				p(q) \not\in \Delta(q)
		\right].
\end{aligned} 
\end{equation*}

\begin{lemma}\label{lemma:number_of_histories}
For every node $q$ it holds that $|\calH^*_q| \leq 2^{n|B|}$.
\end{lemma}
\begin{proof}
It suffices to note that there are at most $2^n$ possible values for $h(p)$ for every node $p$ and to use that there are at most $|B|$ descendants of $q$. For the former claim, if $q$ is the $1$-sink then $h(q)$ can only be $1$. Otherwise, $h(q)$ can only be of the form $\frac{1+\kappa}{l}$ for $l$ an integer between $1$ and $2^n$, so clearly there are only $2^n$ possible values.
\end{proof}

Let $\rho_h(q) = \min(h(q_1),\dots,h(q_k))$. Note that $\rho(q)$ is a random variable whereas $\rho_h(q)$ is a constant. If $H(q) = h$ then the event $\rho(q) = \rho_h(q)$ occurs and thus $p(q) = \round(q,\textsf{min}(\rho_h(q),\hat \rho(q)))$, where $\hat \rho = \median(M_1(q),\dots,M_{n_t}(q))^{-1}$. First, we show that it is enough to focus on $\hat \rho(q)$.

\begin{lemma}\label{lemma:hat_rho}
If $h \in \calH^*_q$ then $(\hat \rho(q) \in \Delta(q)) \land (H(q) = h)$ implies that $(p(q) \in \Delta(q)) \land (H(q) = h)$.
\end{lemma}
\begin{proof}
$H(q) = h$ implies $p(q) = \round(q,\textsf{min}(\rho_h(q),\hat \rho(q)))$.  By Lemma~\ref{lemma:rounding_is_again} it suffices to show that $(\hat \rho \in \Delta(q)) \land (H(q) = h)$ implies $(\textsf{min}(\rho_h(q),\hat \rho(q)) \in \Delta(q)) \land (H(q) = h)$. If $\hat \rho(q) \leq \rho_h(q)$ then the implication holds trivially. If otherwise $\hat \rho(q) > \rho_h(q)$ then $\hat \rho \in \Delta(q)$ implies $\textsf{min}(\rho_h(q),\hat \rho(q)) \leq \hat \rho(q) \leq \frac{1 + \kappa}{|\mods(q)|}$ and, since $h \in \calH^*_q$ guarantees that $\rho_h(q) \geq \frac{1-\kappa}{\max_{j \in [k]} |\mods(q_j)|} \geq \frac{1-\kappa}{|\mods(q)|}$ we conclude that $\textsf{min}(\rho_h(q),\hat \rho(q)) \in \Delta(q)$.
\end{proof}
Let $\mathfrak{M}_{j,h}(q) = \frac{1}{\rho_h(q)n_s}\sum_{r = (j-1) n_s+1}^{j\cdot n_s}|\vZ{r}{h}{q}|$. Lemma~\ref{lemma:not_so_black_magic} implies that the probability of the event $(H(q) = h) \land  (\median_j(M_j(q)) \not\in \nabla(q))$ is upper-bounded by that of $ \median_j(\mathfrak{M}_{j,h}(q)) \not\in \nabla(q)$.
\begin{align*}
\Pr&\left[H(q)=h \text{ and }  p(q) \not\in \Delta(q) \right]
\leq
\Pr\left[H(q)=h \text{ and } \hat \rho(q) \not\in \Delta(q) \right] \tag{Lemma~\ref{lemma:hat_rho}}
\\
&=
\Pr\left[H(q)=h \text{ and } \median(M_1(q),\dots,M_{n_t}(q)) \not\in \nabla(q) \right] \tag{by definition of $\hat{\rho}(q)$}
\\
&\leq \Pr\left[\median(\mathfrak{M}_{1,h}(q),\dots,\mathfrak{M}_{n_t,h}(q)) \not\in \nabla(q)\right]  \tag{Lemma~\ref{lemma:not_so_black_magic}}
\end{align*}
We have now replaced variables by their counterpart in the random process. We use Chebyshev's inequality and  Hoeffding bound to bound $\Pr\left[\median_j(\mathfrak{M}_{j,h}(q)) \not\in \nabla(q)\right]$ from above.

\subsubsection*{Variance upper bound}

By Lemma~\ref{lemma:proba_first_order}, the expected value of $|\vZ{r}{h}{q}|$ is $\mu = \rho_h(q) |\mods(q)|$. Now for the variance,
\begin{align*}
\Va\left[|\vZ{r}{h}{q}| \right] &\leq \Ex\left[|\vZ{r}{h}{q}|^2 \right] = \mu + \sum_{\substack{\alpha,\alpha' \in \mods(q) \\ \alpha \neq \alpha'}} \Pr\left[\alpha \in \vZ{r}{h}{q} \text{ and } \alpha' \in \vZ{r}{h}{q}\right]
\\
&= \mu + \sum_{\substack{\alpha,\alpha' \in \mods(q) \\ \alpha \neq \alpha'}} \Pr\left[\alpha \in \vZ{r}{h}{q} \mid \alpha' \in \vZ{r}{h}{q}\right]\Pr\left[\alpha' \in \vZ{r}{h}{q}\right] 
\\
&\leq \mu + \sum_{\substack{\alpha,\alpha' \in \mods(q) \\ \alpha \neq \alpha'}} \frac{\rho_h(q)^2}{h(\lcpn(\paths(q,\alpha),\paths(q,\alpha'))} \tag{Lemmas~\ref{lemma:proba_first_order} and~\ref{lemma:proba_second_order}}
\end{align*}
Let $\calP = \paths(\alpha,q)$ and $V(\calP) = (q^0_\alpha,q^1_\alpha,q^2_\alpha,\dots,q^{i-1}_\alpha,q)$, with $q^0_\alpha = 1\text{-sink}$. Let $\calP' = \paths(\alpha',q)$ for any $\alpha' \in \mods(q)$ distinct from $\alpha$. Then $\lcpn(\calP,\calP')$ is one of the $q^j_\alpha$. Recall that $I(\alpha,q,j)$ is the set of $\alpha' \in \mods(q)$ such that $\lcpn(\calP,\calP') = q^j_\alpha$.
\begin{align*}
\sum_{\substack{\alpha,\alpha' \in \mods(q) \\ \alpha \neq \alpha'}} \frac{\rho_h(q)^2}{h(\lcpn(\paths(q,\alpha),\paths(q,\alpha'))} = \sum_{\alpha \in \mods(q)} \sum_{j \in [0,i-1]} |I(\alpha,q,j)|\frac{\rho_h(q)^2}{h(q^j_\alpha)} 
\\
\leq \sum_{\alpha \in \mods(q)} \sum_{j \in [0,i-1]} \frac{\rho_h(q)^2 |\mods(q)|}{|\mods(q^j_\alpha)|\cdot h(q^j_\alpha)} \qquad \tag{Lemma~\ref{lemma:derivation_path}}
\end{align*}
Because $h \in \calH^*_q$ and $h(q^j_\alpha)$ is in $\Delta(q^j_\alpha)$ we have $|\mods(q^j_\alpha)|\cdot h(q^j_\alpha) \geq 1 - \kappa$.
\begin{align*}
\sum_{\alpha \in \mods(q)} \sum_{j \in [0,i-1]} \frac{\rho_h(q)^2 |\mods(q)|}{|\mods(q^j_\alpha)|\cdot h(q^j_\alpha)} 
&\leq \frac{\rho_h(q)^2}{1 - \kappa}\sum_{\alpha \in \mods(q)} \sum_{j \in [0,i-1]} |\mods(q)| 
\leq \frac{\mu^2n}{1 - \kappa}
\end{align*}
Putting everything together, we conclude that 
$
\Va[|\vZ{r}{h}{q}|] \leq \mu + \frac{\mu^2n}{1 - \kappa}.
$

\subsubsection*{Median of means}

We have that $\Ex[\mathfrak{M}_{j,h}(q)] = \frac{\mu}{\rho_h(q)} = |\mods(q)|$ and, by independence of the variables $\{\vZ{r}{h}{q}\}_r$ 
$$
\Va[\mathfrak{M}_{j,h}(q)] = \sum_{r = j\cdot n_s+1}^{(j+1)n_s} \frac{\Va[|\vZ{r}{h}{q}|]}{\rho_h(q)^2n_s^2}\leq \frac{1}{\rho_h(q)^2n_s} \left( \mu + \frac{\mu^2n}{1 - \kappa}\right) =   \frac{1}{n_s}\left(\frac{|\mods(q)|}{\rho_h(q)} + \frac{n |\mods(q)|^2}{1-\kappa}\right).
$$
Now, $\mathfrak{M}_{j,h}(q) \in \nabla(q) = \frac{|\mods(q)|}{1 \pm \kappa}$ occurs if and only if $\frac{-\kappa|\mods(q)|}{1 + \kappa} \leq \mathfrak{M}_{j,h}(q) - |\mods(q)| \leq \frac{\kappa|\mods(q)|}{1 - \kappa}$, which is subsumed by $|\mathfrak{M}_{j,h}(q) - |\mods(q)|| \leq \frac{\kappa|\mods(q)|}{1+\kappa}$. So Chebyshev's inequality gives 
\begin{align*}
\Pr\left[\mathfrak{M}_{j,h}(q) \notin \frac{|\mods(q)|}{1 \pm \kappa}\right] &\leq \Pr\left[\big|\mathfrak{M}_{j,h}(q) - |\mods(q)|\big| > \frac{\kappa|\mods(q)|}{1+\kappa}\right] \leq \frac{(1+\kappa)^2}{\kappa^2 |\mods(q)|^2} \Va\left[\mathfrak{M}_{j,h}(q)\right]
\\
&\leq \frac{(1+\kappa)^2}{\kappa^2n_s}\left(\frac{1}{\rho_h(q)|\mods(q)|} + \frac{n}{1-\kappa}\right)
\\
&\leq \frac{(1+\kappa)^2}{\kappa^2n_s}\left(\frac{1}{1 - \kappa} + \frac{n}{1-\kappa}\right) \tag{$\rho_h(q) \geq \frac{1- \kappa}{\max_{j}|\mods(q_j)|} \geq \frac{1- \kappa}{|\mods(q)|}$}
\\
&\leq \frac{2n}{\kappa^2n_s} \tag{$\frac{(1+\kappa)^2}{1-\kappa}$ decreases to  1}
\\
&\leq \frac{1}{4} \tag{$n_s \geq \frac{4n}{\kappa^2}$}
\end{align*}
By taking the median, we decrease the $\frac{1}{4}$ upper bound to a much smaller value. Let $E_j$ be the indicator variable taking value~$1$ if and only if $\mathfrak{M}_{j,h}(q) \not\in \nabla(q)$ and let $\bar E = \sum_{j = 1}^{n_t} E_j$. We have $\Ex[\bar E] \leq \frac{n_t}{4}$ so Hoeffding bound gives
\begin{align*}
\Pr\left[ \median(\mathfrak{M}_{1,h}(q),\dots,\mathfrak{M}_{n_t,h}(q)) \not\in \nabla(q) \right] 
= 
\Pr\left[ \bar E > \frac{n_t}{2}\right] 
\leq 
\Pr\left[ \bar E - \Ex(\bar E) \geq \frac{n_t}{4}\right] 
\leq e^{-n_t/8} \leq \frac{2^{-n|B|}}{16|B|}
\end{align*}
where the last inequality comes from $n_t = 16n|B| \geq 8n|B| + 8\ln(16|B|)$. Using Lemma~\ref{lemma:number_of_histories} we et
\begin{align*}
P(q) \leq \sum_{h \in \calH^*_q} \Pr\left[ \median(\mathfrak{M}_{1,h}(q),\dots,\mathfrak{M}_{n_t,h}(q)) \not\in \nabla(q) \right] \leq \sum_{h \in \calH^*_q} \frac{2^{-n|B|}}{16|B|} \leq  \frac{1}{16|B|}.
\end{align*}
Finally we obtain the lower bound of Lemma~\ref{lemma:proba_p(q)}:
$$
\Pr\left[
		\bigcup_{q \in B}
			p(q) \not\in \Delta(q)
\right] \leq \sum_{q \in B} P(q) \leq \frac{1}{16}
$$

\subsection{Proof of Lemma~\ref{lemma:proba_S(q)}}

We first bound $\Pr\left[\bigcup_{r,q} |S^r(q)| \geq \theta\right]$ from above by 
\begin{align*}
\Pr\left[\bigcup_{r,q} |S^r(q)| \geq \theta\right] &\leq \Pr\left[\bigcup_{r,q} |S^r(q)| \geq \theta \text{ and } \bigcap_{q' \in B} p(q') \in \frac{1 \pm \kappa}{|\mods(q')|}\right] + \Pr\left[\bigcup_{q' \in B} p(q') \not\in \frac{1 \pm \kappa}{|\mods(q')|}\right]
\\
&\leq \Pr\left[\bigcup_{r,q} |S^r(q)| \geq \theta \text{ and } \bigcap_{q' \in B} p(q') \in \Delta(q')\right] + \frac{1}{16} \tag{Lemma~\ref{lemma:proba_p(q)}}
\\
&\leq 
\sum_{r,q}\Pr\left[|S^r(q)| \geq \theta  \text{ and } \bigcap_{q' \in B} p(q') \in \Delta(q')\right] + \frac{1}{16}
\tag{union bound}
\\
&\leq 
\sum_{r,q}\Pr\left[|S^r(q)| \geq \theta  \text{ and } p(q) \in \Delta(q)\right] + \frac{1}{16}
\\
&\leq 
\sum_{r,q}\Pr\left[|S^r(q)|\mathbbm{1}( p(q) \in \Delta(q)) \geq \theta  \right] + \frac{1}{16} \tag{$\theta > 0$}
\\
&\leq 
\sum_{r,q}\Ex\left[|S^r(q)|\mathbbm{1}( p(q) \in \Delta(q))\right] + \frac{1}{16} \tag{Markov's inequality}
\end{align*}
We use the following result to finish the proof of Lemma~\ref{lemma:proba_S(q)}.
\begin{lemma}\label{lemma:white_magic_lemma}
Let $q \in B$ with $\mods(q) \neq \emptyset$. For every $r \in [n_sn_t]$ and $\alpha \in \mods(q)$, $\Ex\left[\frac{\mathbbm{1}(\alpha \in S^r(q))}{p(q)}\right] = 1$.
\end{lemma}
One can see $\Ex\left[\frac{\mathbbm{1}(\alpha \in S^r(q))}{p(q)}\right] = 1$ as a correct variant of ``$\Pr[\alpha \in S^r(q)] = p(q)$'' (which is fallacious as it states that a fixed probability equals a random variable). Let us assume Lemma~\ref{lemma:white_magic_lemma} for now. If $q$ is the $0$-sink then $|S^r(0\text{-sink})| = 0$ so $\Ex\left[|S^r(0\text{-sink})|\mathbbm{1}( p(0\text{-sink}) \in \Delta(0\text{-sink}))\right] = 0$. Now suppose $\mods(q) \neq \emptyset$.
\begin{align*}
\Ex\left[|S^r(q)|\mathbbm{1}( p(q) \in \Delta(q))\right] &= \sum_{\alpha\in \mods(q)} \Ex\left[\mathbbm{1}(\alpha \in S^r(q))\mathbbm{1}( p(q) \in \Delta(q))\right] 
\\
&= \sum_{\alpha\in \mods(q)} \Ex\left[\frac{\mathbbm{1}(\alpha \in S^r(q))}{p(q)}\cdot p(q)\mathbbm{1}( p(q) \in \Delta(q))\right] 
\end{align*}
The random variable $p(q)\mathbbm{1}( p(q) \in \Delta(q))$ is always less than the upper limit of $\Delta(q)$, i.e., $\frac{1 + \kappa}{|\mods(q)|}$. Hence
\begin{align*}
\Ex\left[|S^r(q)|\mathbbm{1}( p(q) \in \Delta(q))\right] &\leq \frac{1 + \kappa}{|\mods(q)|}\sum_{\alpha\in \mods(q)} \Ex\left[\frac{\mathbbm{1}(\alpha \in S^r(q))}{q}\right] 
= 1 + \kappa \tag{Lemma~\ref{lemma:white_magic_lemma}}
\end{align*}
Going back to the bound on $\Pr\left[\bigcup_{r,q} |S^r(q)| \geq \theta\right]$ we get
\begin{align*}
\Pr\left[\bigcup_{r,q} |S^r(q)|  \geq \theta\right] \leq \frac{1}{\theta} \sum_{r,q} (1 +\kappa) + \frac{1}{16} = \frac{(1+\kappa)n_sn_t |B|}{\theta} + \frac{1}{16} \leq \frac{1}{8} \tag{$\theta = \lceil 16(1 +\kappa)n_sn_t|B|\rceil$}
\end{align*}
Thus it only remains to show the correctness of Lemma~\ref{lemma:white_magic_lemma} to finish the proof of Lemma~\ref{lemma:proba_S(q)}.
\begin{proof}[Proof of Lemma~\ref{lemma:white_magic_lemma}]
We proceed by induction on the height of $q$. The base case is that of the $1$-sink. We have that $p(1\text{-sink}) = 1$ and $S^r(1\text{-sink}) = \{\alpha_\emptyset\} = \mods(1\text{-sink})$ so it is indeed verified that 
$$
\Ex\left[\frac{ \mathbbm{1}(\alpha_\emptyset \in S^r(1\text{-sink}))}{p(1\text{-sink})}\right] = 1
$$
For the inductive case, we distinguish the situation where $q$ is a decision from that where $q$ is a $\lor$ node. 

Suppose $q$ is a decision node $ite(x,q_1,q_0)$. Let $b = \alpha(x)$ then and $\alpha_b$ be the restriction of $\alpha$ to $\var(q_b)$.
$$
\frac{\mathbbm{1}(\alpha \in S^r(q))}{p(q)} = \frac{\mathbbm{1}(\alpha \in \reduce(S^r(q_b),\frac{p(q)}{p(q_b)}))}{p(q)}
$$
Let $A$ be the finite set of all possible values for $p(q)$ and $B$ be the finite set of all possible values for $p(q_b$). Let $A(\alpha_b) \subseteq A$ be the set of values $t$ such that $\Pr[p(q) = t,\, \alpha_b \in S^r(q_b)] \neq 0$. Let $B(\alpha_b,t)$ be the set of values $t'$ such that $\Pr[p(q_b) = t',\, p(q) = t,\, \alpha_b \in S^r(q_b)] \neq 0$.
\begin{align*}
\Ex&\left[\frac{\mathbbm{1}(\alpha \in S^r(q))}{p(q)}\right]
= \sum_{t \in A} \frac{1}{t}\Pr\left[\alpha \in S^r(q) ,\,p(q) = t\right]
\\
&= \sum_{t \in A(\alpha_b)} \sum_{t' \in B(\alpha_b,t)} \frac{1}{t}\Pr\left[\alpha_b \in \reduce\left(S^r(q_b),\frac{t}{t'}\right) \bigg|\,\,\begin{matrix} p(q_b) = t' ,\,p(q) = t, \\ \alpha_b \in S^r(q_b) \end{matrix}\right]\Pr\left[\begin{matrix} p(q_b) = t' ,\,p(q) = t, \\ \alpha_b \in S^r(q_b) \end{matrix}\right]
\\
&= \sum_{t \in A(\alpha_b)} \sum_{t' \in B(\alpha_b,t)} \frac{1}{t'}\cdot\Pr\left[p(q_b) = t' ,\,p(q) = t, \, \alpha_b \in S^r(q_b)\right]
\\
&= \sum_{t \in A} \sum_{t' \in B} \frac{1}{t'}\Pr\left[p(q_b) = t' ,\,p(q) = t, \, \alpha_b \in S^r(q_b)\right]
\\
&= \sum_{t' \in B} \frac{1}{t'}\sum_{t \in A} \Pr\left[p(q_b) = t' ,\,p(q) = t, \, \alpha_b \in S^r(q_b)\right] 
\\
&= \sum_{t' \in B} \frac{1}{t'} \Pr\left[p(q_b) = t', \, \alpha_b \in S^r(q_b)\right] = \Ex\left[\frac{\mathbbm{1}(\alpha_b \in S^r(q_b))}{p(q_b)}\right]
\end{align*}
By induction $\Ex\left[\frac{\mathbbm{1}(\alpha_b \in S^r(q_b))}{p(q_b)}\right] = 1$ so we are done.

We move to the case where $q$ is a $\lor$ node. Then 
$$
\frac{\mathbbm{1}(\alpha \in S^r(q))}{p(q)} = \frac{\mathbbm{1}(\alpha \in \reduce(\hat S^r(q_b),\frac{p(q)}{\rho(q)}))}{p(q)}
$$
Let $C$ be the finite set of all possible values for $\rho(q)$ and $C(\alpha,t)$ the set of values $t'$ such that $\Pr[\rho(q) = t',\,p(q) = t,\,\alpha \in \hat S^r(q)] \neq 0$. The same analysis as before with $p(q_b) = t'$ replaced by $\rho(q) = t'$, $\alpha_b \in S^r(q_b)$ replaced by $\alpha \in S^r(q)$, $B$ replaced by $C$ and $B(\alpha_b,t)$ replaced by $C(\alpha,t)$ gives
$$
\Ex\left[ \frac{\mathbbm{1}(\alpha \in S^r(q))}{p(q)} \right] = 
\Ex\left[ \frac{\mathbbm{1}(\alpha \in \hat S^r(q))}{\rho(q)} \right]
$$
There is a unique child $q'$ of $q$ such that $\alpha \in \hat S^r(q)$ if and only if $\alpha \in \reduce(S^r(q'),\frac{\rho(q)}{p(q')})$. Let $D$ be the finite set of all possible values for $p(q')$ and, for $t \in C$, $D(\alpha,t)$ be the set of values such that $\Pr[p(q') = t',\,\rho(q) = t,\,\alpha \in \hat S^r(q)] \neq 0$. Again, the same analysis as before shows that
$$
\Ex\left[ \frac{\mathbbm{1}(\alpha \in \hat S^r(q))}{\rho(q)} \right] = 
\Ex\left[ \frac{\mathbbm{1}(\alpha \in S^r(q'))}{p(q')} \right]
$$
The induction hypothesis gives us $\Ex\left[ \frac{\mathbbm{1}(\alpha \in S^r(q'))}{p(q')} \right] = 1$ so we are done.
\end{proof}
\section{Conclusion}

In this paper, we resolved the open problem of designing an FPRAS for \#nFBDD. Our work also introduces a new technique to quantify dependence, and it would be interesting to extend this technique to other languages that generalize nFBDD. Another promising direction for future work would be to improve the complexity of the proposed FPRAS to enable practical adoption.

\paragraph*{Acknowledgements}
Meel acknowledges the support of the Natural Sciences and Engineering Research Council
of Canada (NSERC), funding reference number RGPIN-2024-05956; de Colnet is supported by the
Austrian Science Fund (FWF), ESPRIT project FWF ESP 235. 
This research was initiated at Dagstuhl Seminar 24171 on ``Automated Synthesis:
Functional, Reactive and Beyond" (https://www.dagstuhl.de/24171). We gratefully acknowledge
the Schloss Dagstuhl – Leibniz Center for Informatics for providing an excellent environment and
support for scientific collaboration. This work was done in part while
de Colnet was visiting the University of Toronto. 

\newcommand{\etalchar}[1]{$^{#1}$}

\clearpage
\appendix

\section{Appendix}

\subsection{Proof of Lemma~\ref{lemma:not_so_black_magic}} 

\notBlackMagic*
\begin{proof}
	The event $H(q) = h$ fixes the values of $p(q')$ to $h(q')$ for every descendant $q'$ of $q$, thus
	\begin{align*}
		\Pr[H(q) = h \text{ and } e(\vec{\mathcal S}(q),\vec{p}(q))] &= \Pr[H(q) = h \text{ and } e(\vec{\mathcal S}(q),\vec{h}(q))]
	\end{align*}
	We rewrite $\approxMCnFBDDcore^*$ in a way that is more convenient. For that we define the following procedures.
	\begin{itemize}
	\item[•] If $q$ is a $\lor$ node
	\begin{itemize}
		\item[--] $\pUnion$ takes in $\vec{S}(c)$ and $p(c)$ for every child $c$ of $q$ and returns $\vec{Z}(q)$. I.e., it executes lines~\ref{line:rho} -- \ref{line:union} in Algorithm~\ref{algorithm:estimateAndSampleOr}.
		\item[--] $\pEstimate_\lor$ takes in $\vec{Z}(q)$ and $p(c)$ for every child $c$ of $q$ (to get $\rho(q)$) and returns $p(q)$. I.e., it executes lines 4 -- \ref{line:round} in Algorithm~\ref{algorithm:estimateAndSampleOr}.
		\item[--] $\pReduce_\lor$ takes in $\vec{Z}(q)$, $p(q)$ and $p(c)$ for every child $c$ of $q$ and returns $\vec{S}(q)$. I.e., it executes lines 8 -- \ref{line:reduce} in Algorithm~\ref{algorithm:estimateAndSampleOr}. 
	\end{itemize} 
	\item[•] If $q$ is a decision node
	\begin{itemize}
	\item[--] $\pEstimate_d$ takes in $p(c)$ for every child $c$ of $q$ and returns $p(q)$. I.e., it executes lines~\ref{line:estimateDecision} in Algorithm~\ref{algorithm:estimateAndSampleDecision}
		\item[--] $\pReduce_d$ takes in $p(q)$, $p(c)$ and $\vec{S}(c)$ for every child $c$ of $q$ and returns $\vec{S}(q)$. I.e., it executes lines 2 -- \ref{line:reduceDecision}  in Algorithm~\ref{algorithm:estimateAndSampleDecision}.
	\end{itemize}
	\end{itemize}
	We make $\approxMCnFBDDcore^*$ record the history $H(q)$ of every node $q$. $H(q)$ is initially empty for every $q$. Before processing $q$, $H(q)$ is computed from the histories and the estimates of $q$'s children by $
	H(q) = compute\_history(H(c), p(c) \mid c \in \children(q))
	$. Recording the history does not modify the distribution of the other variables in $\approxMCnFBDDcore^*$. So $\approxMCnFBDDcore^*$ can be written as follows.\medskip
	
	\begin{algorithm}[H]
		\For{$q \in B$ in order $\prec$}{
			$H(q) = compute\_history(H(c), p(c) \mid c \in \children(q))$\\
			\If{$q$ is a $\lor$ node}{
				$\vec{Z}(q) = \pUnion(\vec{S}(c), p(c) \mid c \in \children(q))$\\ 
				$p(q) = \pEstimate_\lor(\vec{Z}(q), p(c) \mid c \in \children(q))$ \\
				$\vec{S}(q) = \pReduce_\lor(\vec{Z}(q),p(q), p(c) \mid  \in \children(q))$
			}
			\If{$q$ is a decision node}{
				$p(q) = \pEstimate_d(p(c) \mid c \in \children(q))$ \\
				$\vec{S}(q) = \pReduce_d(p(q), p(c),\vec{S}(c) \mid c \in \children(q))$
			}
		}
		\caption{$\approxMCnFBDDcore^*$ summarized}\label{algorithm:estimator_summarized}
	\end{algorithm}\medskip

$\calH_q$ is the finite set of all possible histories for $q$. $T_{q,h}$ is the set of possible values for $p(q)$ under an history $h \in \calH_q$. We introduce sets $\hat S^r_h(q)$ for every history $h \in \calH_q$  and sets $S^r_{h,t}(q)$ for every $h \in \calH_q$ and $t \in T_{q,h}$.\medskip
	
	\begin{algorithm}[H]
		\For{$q \in B$ in order $\prec$}{
			$H(q) = compute\_history(H(c), p(c) \mid c \in \children(q))$\\
			\For{$h \in \calH_q$}{
				\For{$t \in T_{q,h}$}{
					\If{$q$ is a $\lor$ node}{
						$\vec{Z}_h(q) = \pUnion(\vec{S}_{h_c}(c), h(c) \mid c \in \children(q))$\\
						$\vec{S}_{h,t}(q) = \pReduce_\lor(\vec{Z}_h(q),t, h(c) \mid  \in \children(q))$
					}
					\If{$q$ is a decision node}{
						$\vec{S}_{h,t}(q) = \pReduce_d(t, h(c),\vec{S}_{h_c}(c) \mid c \in \children(q))$
					}
				}
			}
						
			\If{$q$ is a $\lor$ node}{
				$p(q) = \pEstimate_\lor(\vec{Z}_{H(c)}(q), p(c) \mid c \in \children(q))$ \\
			}
			\If{$q$ is a decision node}{
				$p(q) = \pEstimate_d(p(c) \mid c \in \children(q))$
			}
		}
		\caption{}\label{algorithm:all_runs}
	\end{algorithm}\medskip
Algorithm~\ref{algorithm:all_runs} plays all possible runs of $\approxMCnFBDDcore^*$ for all possible values of $p(\cdot)$ in lines 1 and 3\,--\,9. To retrieve $\vec S(q)$ for any $q$ from Algorithm~\ref{algorithm:all_runs} it suffices to select $\vec S_{H(q),p(q)}(q)$ and retrieve $\vec Z(q)$ it suffices to select $\vec Z_{H(q)}(q)$. Thus, take $h \in \calH_q$ and let $\vec \calS_h(q)$ be the same as $\vec \calS(q)$ where $\vec Z(q)$ is replaced by $\vec Z_{h}(q)$ and where each where each $\vec S(q')$ and $\vec Z(q')$ for $q'$ a descendant of $q$ are replaced by $\vec S_{h_{q'},h(q')}(q')$ and $\vec Z_{h_{q'}}(q')$ with $h_{q'}$ the restriction of $h$ to the descendants of $q'$, then
	\begin{align*}
		\Pr[H(q) = h \text{ and } e(\vec{S}(q),\vec h(q))] = \Pr[H(q) = h \text{ and } e(\vec{\calS}_h(q),\vec h(q))] 
	\end{align*}
	If in Algorithm~\ref{algorithm:all_runs} we remove line 2 and lines 10--13 then we do not change the distribution of the variables $S^r_{h,t}(q)$ and $\hat S^r_{h}(q)$. Removing these lines transforms Algorithm~\ref{algorithm:all_runs} into the $\mathfrak{S}$-process, thus
	\begin{align*}
		\Pr[e(\vec{\calS}_{h}(t),\vec h(q))] = 
		\Pr[e(\vec{\mathfrak{S}}_h(q),\vec h(q))]
	\end{align*}
	It follows that
	\begin{align*}
		\Pr[H(v) = h \text{ and } e(\vec{\calS}(q),\vec{p}(q))] &= \Pr[H(v) = h \text{ and } e(\vec{\cal S}_h(q),\vec{h}(q))] 
		\\
		&\leq \Pr[e(\vec{\cal S}_h(q),\vec{h}(q))]
		= \Pr[e(\vec{\mathfrak{S}}_h(q),h,\vec h(q))]
	\end{align*}
\end{proof}

\subsection{Proofs of Lemmas~\ref{lemma:proba_first_order} and~\ref{lemma:proba_second_order}}

\probaFirstOrder*
\begin{proof}
Let $q \in L_i$. We proceed by induction on $i$. The base case $i = 0$ is immediate since $\vY{r}{h_\emptyset}{1}{1\text{-sink}} = \{\alpha_\emptyset\}$ and $\vY{r}{h_\emptyset}{\emptyset}{0\text{-sink}} = \emptyset$ are the only variables for $L_0$ (and $\frac{1}{\infty} = 0$ by definition). Now let $i > 0$, $q \in L_i$, and suppose that the statement holds for all $\vY{r}{h'}{t'}{q'}$ and $\alpha' \in \mods(q')$ with $q' \in L_{< i}$. If $i$ is odd then $q$ is a decision node $ite(x,q_1,q_0)$ with $q_0$ and $q_1$ in $L_{i-1}$ (because $B$ is alternating). Let $b = \alpha(x)$ and let $\alpha'$ be the restriction of $\alpha$ to $var(\alpha) \setminus \{x\}$. Then 
\begin{align*}
\Pr\left[\alpha \in \vY{r}{h}{t}{q}\right] 
&= \Pr\left[\alpha \in \reduce\left(\vY{r}{h_b}{t_b}{q_b},\frac{t}{t_b}\right) \otimes \{x \mapsto b\}\right] 
\\
&= \Pr\left[\alpha' \in \vY{r}{h_b}{t_b}{q_b}\right]\Pr\left[\alpha' \in \reduce\left(\vY{r}{h_b}{t_b}{q_b},\frac{t}{t_b}\right) \Big|\, \alpha' \in \vY{r}{h_b}{t_b}{q_b} \right]
\\
&= t_b \cdot \frac{t}{t_b} = t \tag{by induction}
\end{align*}
Now if $i$ is even then $q$ is a $\lor$-node with children $\children(q) = (q_1,\dots,q_k)$ all in $L_{i-1}$. Let $j$ be the smallest integer such that $\alpha \in \mods(q_j)$ and let $t_{\min} = \min(h(q_1),\dots,h(q_k))$. Then 
\begin{align*}
\Pr\left[\alpha \in \vZ{r}{h}{q}\right] 
&= \Pr\left[\alpha \in \reduce\left(\vY{r}{h_j}{t_j}{q_j},\frac{t_{\min}}{t_j}\right)\right] 
\\
&= \Pr\left[\alpha \in \vY{r}{h_j}{t_j}{q_j}\right]\Pr\left[\alpha \in \reduce\left(\vY{r}{h_j}{t_j}{q_j},\frac{t_{\min}}{t_j}\right) \Big|\, \alpha \in \vY{r}{h_j}{t_j}{q_j} \right]
\\
&= t_j \cdot \frac{t_{\min}}{t_j} = t_{\min} \tag{by induction}
\end{align*}
And for $t \leq t_{\min}$ we have that 
\begin{align*}
\Pr\left[\alpha \in \vY{r}{h}{t}{q}\right] 
&= \Pr\left[\alpha \in \reduce\left(\vZ{r}{h}{q},\frac{t}{t_{\min}}\right)\right]
\\
&= \Pr\left[\alpha \in \vZ{r}{h}{q}\right]\Pr\left[\alpha \in \reduce\left(\vZ{r}{h}{q},\frac{t}{t_{\min}}\right) \Big|\, \alpha \in \vZ{r}{h}{q}\right]
\\
&= t_{\min} \cdot \frac{t}{t_{\min}} = t
\end{align*}
\end{proof}

\probaSecondOrder*
\begin{proof}
We are going to prove a stronger statement, namely, that for every $i$, for every $q$ and $q'$ (potentially $q = q'$) in $L_i$ and every $t$ and $t'$ such that $t = t'$ when $q = q'$,  and every $\alpha \in \mods(q)$ and $\alpha' \in \mods(q')$, and every compatible histories $h$ and $h'$ for $q$ and $q'$, respectively, we have that
\begin{equation}\label{equation:stronger_statement}
\Pr\left[\alpha \in \vY{r}{h}{t}{q} \text{ and } \alpha' \in \vY{r}{h'}{t'}{q'}\right] \leq \frac{tt'}{t^*}.
\end{equation}
where $t^* = t$ if $(q,\alpha) = (q',\alpha')$ and $t^* = h(\lcpn(\paths(\alpha,q),\paths(\alpha',q'))$ otherwise. In addition
if $i$ is even, so $q$ and $q'$ are $\lor$-nodes, then
\begin{equation}\label{equation:stronger_statement_2}
\Pr\left[\alpha \in \vZ{r}{h}{q} \text{ and } \alpha' \in \vZ{r}{h'}{q'}\right] \leq \frac{t_{\min}t'_{\min}}{t^*}.
\end{equation}
where $t_{\min} = \min(h(c) \mid c \in \children(q))$ and $t'_{\min} = \min(h'(c) \mid c \in \children(q'))$ and $t^* = t_{\min}$ if $(q,\alpha) = (q',\alpha')$ and $t^* = h(\lcpn(\paths(\alpha,q),\paths(\alpha',q'))$ otherwise.

The inequalities (\ref{equation:stronger_statement}) and  (\ref{equation:stronger_statement_2}) are straightforward when $(\alpha,q) = (\alpha',q')$ because then $h = h'$ (by compatibility) and $t = t' = t^*$ or $t_{\min} = t'_{\min} = t^*$ and we can use Lemma~\ref{lemma:proba_first_order}. In particular~(\ref{equation:stronger_statement}) holds when $q = q' = 1\text{-sink}$. Now we assume $(\alpha, q) \neq (\alpha',q')$ and proceed by induction on $i$. The base case $i = 0$ holds true by the previous remark (note that neither $q$ nor $q'$ can be the $0$-sink because $\mods(q)$ and $\mods(q')$ must not be empty). 

\begin{itemize}
\item \textbf{Case $i$ odd}. In this case $q$ and $q'$ are decision nodes. Let $q = ite(x,q_1,q_0)$ and $q' = ite(y,q'_1,q'_0)$. Then $h = h_0 \cup h_1 \cup \{q_0 \mapsto t_0, q_1 \mapsto t_1\}$ for some compatible histories $h_0$ and $h_1$ for $q_0$ and $q_1$, respectively, and $t_0$ and $t_1$ such that $t = (\frac{1}{t_0} + \frac{1}{t_1})^{-1}$. Similarly, $h' = h'_0 \cup h'_1 \cup \{q'_0 \mapsto t'_0, q'_1 \mapsto t'_1\}$. Let $b = \alpha(x)$ and $c = \alpha'(y)$. Let also $\beta$ be the restriction of $\alpha$ to $var(\alpha) \setminus \{x\}$ and $\beta'$ be the restriction of $\alpha'$ to $var(\alpha') \setminus \{y\}$. Then 
\begin{align*}
\Pr&[\alpha \in \vY{r}{h}{t}{q} \text{ and } \alpha' \in \vY{r}{h'}{t'}{q'}] 
\\
&= \Pr\left[\beta \in \reduce\left(\vY{r}{h_b}{t_b}{q_b},\frac{t}{t_b}\right), \beta' \in \reduce\left(\vY{r}{h'_c}{t'_c}{q'_c},\frac{t'}{t'_c}\right)  \right.\\
& \quad\qquad \Big|\,\beta \in \vY{r}{h_b}{t_b}{q_b}, \beta' \in \vY{r}{h'_c}{t'_c}{q'_c}\bigg] \Pr\left[\beta \in \vY{r}{h_b}{t_b}{q_b} \text{ and } \beta' \in \vY{r}{h'_c}{t'_c}{q'_c}\right]
\end{align*} 
Now, because $q$ and $q'$ are both in $L_i$, neither is an ancestor of the other and thus the two $\reduce$ are independent: the output of one $\reduce$ does not modify the set fed into the second $\reduce$ nor its output. Thus the probability becomes 
\begin{multline*}
\Pr\left[\beta \in \reduce\left(\vY{r}{h_b}{t_b}{q_b},\frac{t}{t_b}\right)\Big|\,\beta \in \vY{r}{h_b}{t_b}{q_b}\right]
\\
\qquad\qquad\qquad \cdot 
\Pr\left[\beta' \in \reduce\left(\vY{r}{h'_c}{t'_c}{q'_c},\frac{t'}{t'_c}\right) \Big|\,\beta' \in \vY{r}{h'_c}{t'_c}{q'_c}\right]
\\
\cdot 
\Pr\left[\beta \in \vY{r}{h_b}{t_b}{q_b} \text{ and } \beta' \in \vY{r}{h'_c}{t'_c}{q'_c}\right]
\end{multline*} 
which is equal to $\frac{tt'}{t_bt'_c}\Pr\left[\beta \in \vY{r}{h_b}{t_b}{q_b} \text{ and } \beta' \in \vY{r}{h'_c}{t'_c}{q'_c}\right]$. Now, if $(\beta,q_b) = (\beta',q'_c)$ then $t_b = t'_c$ (because $h$ and $h'$ are compatible) and $\Pr\left[\beta \in \vY{r}{h_b}{t_b}{q_b} \text{ and } \beta' \in \vY{r}{h'_c}{t'_c}{q'_c}\right]$ $= \Pr\left[\beta \in \vY{r}{h_b}{t_b}{q_b}\right] = t_b = t'_c$ by Lemma~\ref{lemma:proba_first_order}. So $\Pr[\alpha \in \vY{r}{h}{t}{q} \text{ and } \alpha' \in \vY{r}{h'}{t'}{q'}] \leq  \frac{tt'}{t_b} = \frac{tt'}{h(q_b)}$. By assumption, $(\alpha,q) \neq (\alpha',q')$, if $q \neq q'$ then the two derivation paths $\paths(\alpha,q)$ and $\paths(\alpha',q')$ diverge for the first time at $q_b$, and if $q = q'$ then $x = y$ and $\alpha(x) = 1 - \alpha'(x)$ (because $(\alpha,q) \neq (\alpha',q')$ by assumption). In this case the derivation paths still diverge for the first time at $q_b$: one follows the $0$-edge and the other follows the $1$-edge. So in both cases $\lcpn(\paths(\alpha,q),\paths(\alpha',q')) = q_b$ and we are done. We still have $(\beta,q_b) \neq (\beta',q'_c)$ to consider. In this case the paths $\paths(\beta,q_b)$ and $\paths(\beta',q'_c)$ diverge for the first time at some node $q^*$ below $q_b$ and $q'_c$ so by induction $\Pr\left[\beta \in \vY{r}{h_b}{t_b}{q_b} \text{ and } \beta' \in \vY{r}{h'_c}{t'_c}{q'_c}\right] \leq t_b t'_c/h_b(q^*) = t_b t'_c/h(q^*)$. So $\Pr[\alpha \in \vY{r}{h}{t}{q} \text{ and } \alpha' \in \vY{r}{h'}{t'}{q'}] \leq tt'/h(q^*)$. But $q^*$ is also the first node where $\paths(\alpha,q)$ and $\paths(\alpha',q')$ diverge, hence the result.

\item \textbf{Case $i$ even.} In this case $q$ and $q'$ are both $\lor$-nodes. Say $q = q_1 \lor \dots \lor q_k$ and $q' = q'_1 \lor \dots \lor q'_m$. Then $h = h_1 \cup \dots \cup h_k \cup \{q_1 \mapsto t_1,\dots,q_k \mapsto t_k\}$ and  $h' = h'_1 \cup \dots \cup h'_k \cup \{q'_1 \mapsto t'_1,\dots,q'_m \mapsto t'_m\}$. Let $t_{\min} = \min(t_1,\dots,t_k)$ and $t'_{\min} = \min(t'_1,\dots,t'_m)$.
\begin{align*}
&\Pr\left[\alpha \in \vY{r}{h}{t}{q} \text{ and } \alpha' \in \vY{r}{h'}{t'}{q'}\right] 
\\
= &\Pr\left[\alpha \in \reduce\left(\vZ{r}{h}{q}, \frac{t}{t_{\min}}\right) \text{ and } \alpha' \in \reduce\left(\vZ{r}{h'}{q'}, \frac{t'}{t'_{\min}}\right)\right] 
\\
= &\Pr\left[\alpha \in \reduce\left(\vZ{r}{h}{q}, \frac{t}{t_{\min}}\right) ,\, \alpha' \in \reduce\left(\vZ{r}{h'}{q'}, \frac{t'}{t'_{\min}}\right) \Big| \alpha \in \vZ{r}{h}{q},\, \alpha' \in \vZ{r}{h'}{q'}\right]\\
&\cdot \Pr\left[\alpha \in \vZ{r}{h}{q} \text{ and } \alpha' \in \vZ{r}{h'}{q'}\right]
\end{align*} 
The $\reduce$ events are independent because $q$ and $q'$ both belong to $L_i$ and thus neither in an ancestor of the other: the output of the first $\reduce$ does not influence the output of the second one, even with the knowledge that $\alpha \in \vZ{r}{h}{q}$ and $\alpha' \in \vZ{r}{h'}{q'}$. So the probability becomes
\begin{multline*}
\Pr\left[\alpha \in \reduce\left(\vZ{r}{h}{q}, \frac{t}{t_{\min}}\right)  \Big| \alpha \in \vZ{r}{h}{q}\right] 
\\
\qquad\qquad\qquad\cdot \Pr\left[\alpha' \in \reduce\left(\vZ{r}{h'}{q'}, \frac{t'}{t'_{\min}}\right)  \Big| \alpha' \in \vZ{r}{h'}{q'}\right] 
\\
\cdot \Pr\left[\alpha \in \vZ{r}{h}{q} \text{ and } \alpha' \in \vZ{r}{h'}{q'}\right]
\end{multline*}  
which is $\frac{tt'}{t_{\min}t'_{\min}}\Pr\left[\alpha \in \vZ{r}{h}{q} \text{ and } \alpha' \in \vZ{r}{h'}{q'}\right]$. Now, there are a unique $j$ and $\ell$ such that $\alpha \in \vZ{r}{h}{q}$ only if $\alpha \in \vY{r}{h_j}{t_j}{q_j}$ and $\alpha' \in \vZ{r}{h'}{q'}$ only if $\alpha' \in \vY{r}{h'_\ell}{t'_\ell}{q'_\ell}$. Thus
\begin{align*}
\Pr\left[\alpha \in \right.&\left. \vZ{r}{h}{q} \text{ and } \alpha' \in \vZ{r}{h'}{q'}\right]
\\
&= \Pr\left[\alpha \in \reduce\left( \vY{r}{h_j}{t_j}{q_j}, \frac{t_{\min}}{t_j}\right) \text{ and } \alpha' \in \reduce\left(\vY{r}{h'_\ell}{t'_\ell}{q'_\ell}, \frac{t'_{\min}}{t'_\ell}\right)\right].  
\end{align*} 
$q_j$ and $q'_\ell$ belong to the same layer so with similar arguments we find that 
\begin{align*}
\Pr\left[\alpha \in \vZ{r}{h}{q} \text{ and } \alpha' \in \vZ{r}{h'}{q'}\right]
= \frac{t_{\min}t'_{\min}}{t_j t'_l}\Pr\left[\alpha \in \vY{r}{h_j}{t_j}{q_j} \text{ and } \alpha' \in \vY{r}{h'_\ell}{t'_\ell}{q'_\ell}\right].
\end{align*}
It is possible that $(\alpha,q_j) = (\alpha',q'_\ell)$ but then $q \neq q'$ for otherwise we would have $(\alpha,q) = (\alpha',q')$, against assumption. In the case $(\alpha,q_j) = (\alpha',q'_\ell)$ we use Lemma~\ref{lemma:proba_first_order} and find $\Pr\left[\alpha \in \vY{r}{h_j}{t_j}{q_j} \text{ and } \right.$ $\left.\alpha' \in \vY{r}{h'_\ell}{t'_\ell}{q'_\ell}\right] = \Pr\left[\alpha \in \vY{r}{h_j}{t_j}{q_j}\right] = t_j = t'_\ell$. So $\Pr\left[\alpha \in \vZ{r}{h}{q} \text{ and } \alpha' \in \vZ{r}{h'}{q'}\right] = \frac{t_{\min}t'_{\min}}{t_j}$ and  $\Pr\left[\alpha \in \vY{r}{h}{t}{q} \text{ and } \alpha' \in \vY{r}{h'}{t'}{q'}\right] = \frac{tt'}{t_j}$. When $(\alpha,q_j) = (\alpha',q'_\ell)$, the paths $\paths(\alpha,q)$ and $\paths(\alpha',q')$ diverge for the first time at $q_j = q'_\ell$. So $t_j = h(\lcpn(\paths(\alpha,q),\paths(\alpha',q')))$ and we are done. Now let us assume that $(\alpha,q_j) \neq (\alpha',q'_\ell)$, then we use the induction hypothesis and, denoting $q^* = \lcpn(\paths(\alpha,q_j),\paths(\alpha',q'_\ell))$, we have 
\begin{align*}
\Pr\left[\alpha \in \vZ{r}{h}{q} \text{ and } \alpha' \in \vZ{r}{h'}{q'}\right]
\leq \frac{t_{\min}t'_{\min}}{h(q^*)}
\end{align*}
and
\begin{align*}
\Pr\left[\alpha \in \vY{r}{h}{t}{q} \text{ and } \alpha' \in \vY{r}{h'}{t'}{q'}\right] \leq \frac{tt'}{h(q^*)}.
\end{align*}
When $(\alpha,q_j) \neq (\alpha',q'_\ell)$, the first node where $\paths(\alpha,q)$ and $\paths(\alpha',q')$ diverge is also the first node where $\paths(\alpha,q_j)$ and $\paths(\alpha',q'_\ell)$ diverge, so $q^*$.
\end{itemize}
This finishes the proof of the inductive case.
\end{proof}

\end{document}